\documentclass[11pt,reqno]{amsart}
\newtheorem{theorem}{Theorem}[section]
\newtheorem{proposition}[theorem]{Proposition}
\newtheorem{corollary}[theorem]{Corollary}
\newtheorem{remark}[theorem]{Remark}
\newtheorem{lemma}[theorem]{Lemma}
\newtheorem{example}[theorem]{Example}
\newtheorem{definition}[theorem]{Definition}

\usepackage{appendix}
\usepackage{amssymb}
\usepackage{amscd}
\usepackage{mathrsfs}
\usepackage{graphicx}
\usepackage{amsmath}
\usepackage{amsfonts}
\usepackage{color}
\usepackage{bm}
\usepackage{bbm}
\usepackage[all]{xy}
\usepackage[latin1]{inputenc} 
\usepackage{graphicx} 
\numberwithin{equation}{section}
\usepackage{hyperref}
\allowdisplaybreaks[4]
\begin{document}
	\title{$n$-Lie conformal algebras and its associated infinite-dimensional $n$-Lie algebras}
	\author{Mengjun Wang$^{1}$, Lipeng Luo$^{2}$ and Zhixiang Wu$^{3}$}
	\address{\textsuperscript{1}School of  Mathematical  Sciences, Zhejiang University, Hangzhou, 310027, P.R.China.} 
	\address{\textsuperscript{2}School of Mathematical Sciences, Tongji University, Shanghai, 200092, P.R.China.} 
	\address{\textsuperscript{3}School of  Mathematical  Sciences, Zhejiang University, Hangzhou, 310027, P.R.China.} 
	\email{\textsuperscript{1}wangmengjun@zju.edu.cn}
	\email{\textsuperscript{2}luolipeng@tongji.edu.cn}
	\email{\textsuperscript{3}wzx@zju.edu.cn}
	\thanks{This work was supported by the National Natural Science Foundation of China (No. 11871421, 12171129) and the Zhejiang Provincial Natural Science Foundation of China (No. LY20A010022) and the Fundamental Research Funds for the Central Universities (No. 22120210554).}
	\subjclass[2010]{17B10; 17B56; 81R50; 81T40}

\keywords{$n$-Lie conformal algebras, Lie conformal algebras, Filippov algebras, cohomology}
\footnote{The second author is the corresponding author: Lipeng Luo (luolipeng@tongji.edu.cn)}
\begin{abstract}
In this paper, we introduce a $\{\lambda_{1\to n-1}\}$-bracket and a distribution notion of an $n$-Lie conformal algebra. For any $n$-Lie conformal algebra $R$, there exists a series of associated infinite-dimensional linearly compact $n$-Lie algebras $\{(\mathscr{L}ie_p\mbox{ }R)_\_\}_{(p\ge1)}$. We show that torsionless finite $n$-Lie conformal algebras $R$ and $S$ are isomorphic if and only if $(\mathscr{L}ie_p\mbox{ }R)_\_\simeq (\mathscr{L}ie_p\mbox{ }S)_\_$ as linearly compact $n$-Lie algebras with $\partial_{t_i}$-action for any $p\ge1$. Moreover, the representation and cohomology theory of $n$-Lie conformal algebras are established. In particular, the complex of $R$ is isomorphic to a subcomplex of $n$-Lie algebra $(\mathscr{L}ie_p\mbox{ }R)_\_$.
	\end{abstract}
\maketitle
\baselineskip=16pt
	\section{Introduction}
	For an element $x$ of a Lie algebra $(\mathfrak{g},[\cdot,\cdot])$, denote $ad_x$ as the induced map $ad_x(y)=[x,y],y\in \mathfrak{g}$. Then the Jacobi identity condition is equivalent to the condition that $ad_x$ are derivations of $\mathfrak{g}$. As a generalization of Lie algebra, an $n$-Lie algebra is a space $\mathfrak{g}$ endowed with an $n$-linear map $[\cdot,\cdot,\cdots,\cdot]:\wedge^n \mathfrak{g}\to\mathfrak{g}$ satisfying every $ad_X$ is a derivation of $\mathfrak{g}$ for $X=x_1\wedge\cdots\wedge x_{n-1}\in \wedge^{n-1} \mathfrak{g}$, where $ad_X(y)=[x_1,\cdots,x_{n-1},y],y\in \mathfrak{g}$.	This notion is derived from ternary Lie algebras introduced by Nambu \cite{y} to generalize the classical Hamiltonian mechanics. In \cite{f}, Filippov introduced the theory of $n$-Lie algebras, which is the algebraic structure corresponding to Nambu mechanics. Hence $n$-Lie algebras are also called Filippov algebras. So far, the theory of $n$-Lie algebras has been developed by \cite{ai2,ct,t}, and widely used in many areas of both mathematics and physics \cite{bgs,h,st,v}. A more detailed review can be seen in \cite{ai1}.

On the other hand, from a Lie algebra $\mathfrak{g}$ one can obtain a Lie conformal algebra $Cur\mathfrak{g}$, and from a Lie conformal algebra $R$ one can get an associated infinite-dimensional Lie algebra $(Lie\ R)_{\_}$, which is called the annihilation algebra of $R$ \cite{ak}. Conformal algebras were introduced by Kac \cite{k} as an algebraic tool for studying properties of the operator product expansion (OPE) in two-dimensional conformal field theory. Kac et. al. further developed the structure theory \cite{ak} and cohomology theory \cite{bkv} of conformal algebras, and generalized them to pseudoalgebras \cite{bdk}. Based on Kac's fundamental works, a lot of related researchs have emerged \cite{bk,lhw,s}. In this article, we establish the structure and cohomology theory of $n$-Lie conformal algebras which are an amalgamation of $n$-Lie algebras and conformal algebras, and reveal the relation of $n$-Lie conformal algebras and their associated $n$-Lie algebras.
	
	The paper is organized as follows. In section \ref{s2}, we introduce the $\{\lambda_{1\to n-1}\}$-brackets and distribution notion of an $n$-Lie conformal algebra.  Some concrete examples are given. We find that there has no non-commutative $n$-Lie conformal algebras of rank one when $n\ge3$. But there are some non-commutative  $n$-Lie conformal algebras of rank two. In section \ref{s3}, we construct the annihilation algebras $(\mathscr{L}ie_p\mbox{ }R)_\_$ of an $n$-Lie conformal algebra $R$, and study stucture and properties of these $n$-Lie algebras. We show that an isomorphism between annihilation algebras determines an isomorphism between torsionless finite $n$-Lie conformal algebras. The representation and basic cohomology of $n$-Lie conformal algebras are given in section \ref{s4}. And we show not only every $R$-module carries a $(\mathscr{L}ie_p\mbox{ }R)_\_$-module structure, but also the complex of $R$ is isomorphic to a subcomplex of the $n$-Lie complex of $(\mathscr{L}ie_p\mbox{ }R)_\_$. 
	
	Throughout this paper, we work over the complex field $\mathbb{C}$. For the sake of brevity and readability, we sometimes use symbols with index $1\to m$ to mean enumerating the objects from $1$ to $m$. For example, $\lambda_{1\to n-1}$ means $\lambda_1,\lambda_2,\cdots,\lambda_{n-1}$ and $k^{1\to p}_{1\to q}$ means $k^1_1,k^1_2,\cdots,k^1_q,\cdots,k^p_1,\cdots,k^p_q$. If there is an arithmetic sign before the symbols with index $1\to m$, it means the operation is carried out on every enumerated item. $\lambda+\lambda_{1\to n-1}$, for instance, means $\lambda+\lambda_1+\cdots+\lambda_{n-1}$. 
	
	\section{$n$-Lie conformal algebras}\label{s2}
	In this section, we introduce the two equivalent definitions of $n$-Lie conformal algebras. The distribution notion has advantages in attaching to $n$-Lie algebras, while it is more convenient to compute with $\{\lambda_{1\to n-1}\}$-bracket. Although non-abelian Lie conformal algebras of rank one, i.e. Virasoro conformal algebras \cite{ak}, are of significance to the study of Lie conformal algebras, we find that there has no non-commutative $n$-Lie conformal algebras of rank one when $n\ge3$. Then the question arises in mind whether all $n$-Lie conformal algebras with rank lower than $n-1$ are commutative. We negate it by classifying a special class of $n$-Lie conformal algebras of rank two.

	\begin{definition}
		\begin{em}
			An \textit{$n$-Lie conformal algebra} is a $\mathbb{C}[\partial]$-module R equipped with a $\{\lambda_{1\to n-1}\}$-bracket satisfying axioms (C1)-(C4).
			
			(C1) $[{a^1} _{\lambda_1}{a^2}_{\lambda_2}{\cdots}_{\lambda_{n-1}}a^n]\in \mathbb{C}[\partial,\lambda_1,\lambda_2,\cdots,\lambda_{n-1}]\otimes_{\mathbb{C}[\partial]}R$,
			
			(C2) $[{a^1} _{\lambda_1}{\cdots}_{\lambda_{i-1}}{\partial a^i}_{\lambda_i}{\cdots}_{\lambda_{n-1}}a^n]=-\lambda_i[{a^1} _{\lambda_1}{a^2}_{\lambda_2}{\cdots}_{\lambda_{n-1}}a^n]$ for $i=1,2,\cdots,n-1$,
			
			\qquad\mbox{ }$[{a^1} _{\lambda_1}{\cdots}_{\lambda_{n-2}}{a^{n-1}}_{\lambda_{n-1}}\partial a^n]=(\partial+\lambda_1+\lambda_2+\cdots+\lambda_{n-1})[{a^1} _{\lambda_1}{a^2}_{\lambda_2}{\cdots}_{\lambda_{n-1}}a^n]$,
			
			(C3) $[{a^1} _{\lambda_1}\cdots{ a^i}_{\lambda_i}{ a^{i+1}}_{\lambda_{i+1}}{\cdots}_{\lambda_{n-1}}a^n]=-[{a^1} _{\lambda_1}\cdots{ a^{i+1}}_{\lambda_{i+1}}{ a^i}_{\lambda_i}{\cdots}_{\lambda_{n-1}}a^n]$ for $i=1,2,\cdots,n-2$,
			
			\qquad\mbox{ }$[{a^1} _{\lambda_1}\cdots{ a^{n-1}}_{\lambda_{n-1}}a^n]=-[{a^1} _{\lambda_1}\cdots{ a^n}_{-\partial-\lambda_1-\cdots-\lambda_{n-1}}a^{n-1}]$,
			
			(C4) $\sum\limits_{i=1}^{n}(-1)^{n-i}[{a^1} _{\lambda_1}\cdots{ a^{i-1}}_{\lambda_{i-1}}{ a^{i+1}}_{\lambda_{i+1}}{\cdots}_{\lambda_{n-1}}{a^n}_{\lambda_n}[{ a^i}_{\lambda_i}{b^2}_{\lambda_{n+1}}{\cdots}_{\lambda_{2n-2}}b^n]]$
			
			\qquad$=[[{a^1} _{\lambda_1}{a^2}_{\lambda_2}{\cdots}_{\lambda_{n-1}}a^n]_{\lambda_1+\cdots+\lambda_n}{b^2}_{\lambda_{n+1}}{\cdots}_{\lambda_{2n-2}}b^n]$,\\
			for all $a^1,a^2,\cdots,a^n,b^2,\cdots,b^n \in R$. An $n$-Lie conformal algebra $R$ is said to be \textit{finite} if $R$ is a finitely generated $\mathbb{C}[\partial]$-module. The \textit{rank} of $n$-Lie conformal algebra $R$ is its rank as $\mathbb{C}[\partial]$-module.
		\end{em}	
	\end{definition}
	
A $2$-Lie conformal algebra is simply called Lie conformal algebra.
\begin{remark}
	Let $\{e_i\}$ be a set of generators of the $\mathbb{C}[\partial]$-module $ R $. Then the $\{\lambda_{1\to n-1}\}$-bracket of $R$ can be expressed as $$[{e_{i_1}} _{\lambda_1}{e_{i_2}}_{\lambda_2}{\cdots}_{\lambda_{n-1}}e_{i_n}]=\sum_{k}Q^{i_1,i_2,\cdots,i_n}_k(\lambda_1,\lambda_2,\cdots,\lambda_{n-1},\partial)e_k.$$ The $n$-Lie pseudoalgebra bracket on $R$ is defined by $$[e_{i_1}*e_{i_2}*\cdots*e_{i_n}]=\sum_{k}P^{i_1,i_2,\cdots,i_n}_k(\partial\otimes1\otimes\cdots\otimes1,1\otimes\partial\otimes1\otimes\cdots\otimes1,\cdots,1\otimes\cdots\otimes1\otimes\partial)e_k,$$ where $P^{i_1,i_2,\cdots,i_n}_k(x_1,x_2,\cdots,x_n)=Q^{i_1,i_2,\cdots,i_n}_k(-x_1,-x_2,\cdots,x_1+x_2+\cdots+x_n)$.
	Replacing the $\sigma_i$ of Filippov identity defined in \cite{sw} by $\sigma_i^{-1}$, we obtain the corresponding $n$-Lie $H$-pseudoalgebra structure on $R$.
\end{remark}
 Obviously, we can define subalgebras, ideals, quotients, simple algebras and homomorphisms of $n$-Lie conformal algebras.

For $a^1,a^2,\cdots,a^n\in R$, set
	\begin{eqnarray}
	[{a^1} _{\lambda_1}{a^2}_{\lambda_2}{\cdots}_{\lambda_{n-1}}a^n]=\sum_{k_1,\cdots,k_{n-1}\in\mathbb{N}_+}{\lambda_1}^{(k_1)}\cdots{\lambda_{n-1}}^{(k_{n-1})}{a^1}_{(k_1)}{a^2}_{(k_2)}{\cdots}_{(k_{n-1})}a^n,
	\end{eqnarray} where ${\lambda_i}^{(k_i)}:=\frac{{\lambda_i}^{k_i}}{{k_i}!}$. Then we can get a family of $n$-linear products $\{{\bullet}_{(k_1)}{\bullet}_{(k_2)}{\cdots}_{(k_{n-1})}{\bullet}\}$ on $R$ and rephrase the above $n$-Lie conformal equivalent axioms.
	
	(C$1^\prime$) $ {a^1}_{(k_1)}{a^2}_{(k_2)}{\cdots}_{(k_{n-1})}a^n=0 $ for $k_1+k_2+\cdots+k_{n-1}\gg0$,
	
	(C2$^\prime$) ${a^1}_{(k_1)}{\cdots}_{(k_{i-1})}{\partial a^i}_{(k_i)}{\cdots}_{(k_{n-1})}a^n=-k_i{a^1}_{(k_1)}{\cdots}_{(k_{i-1})}{ a^i}_{(k_i-1)}{\cdots}_{(k_{n-1})}a^n$ for $i=1,\cdots,n-1$,
	
	\qquad\mbox{ }${a^1}_{(k_1)}{\cdots}_{(k_{n-1})}{\partial a^n}=\partial({a^1}_{(k_1)}{a^2}_{(k_2)}{\cdots}_{(k_{n-1})}a^n)+\sum_{i=1}^{n-1}k_i{a^1}_{(k_1)}{\cdots}_{(k_{i-1})}{ a^i}_{(k_i-1)}{\cdots}_{(k_{n-1})}a^n$,
	
	(C$3^\prime$) ${a^1}_{(k_1)}{a^2}_{(k_2)}{\cdots}_{(k_{n-1})}a^n=-{a^1}_{(k_1)}\cdots{a^{i+1}}_{(k_{i+1})}{a^i}_{(k_i)}{\cdots}_{(k_{n-1})}a^n$ for $i=1,\cdots,n-1$,
	
	\qquad$=-\sum\limits_{j,j_{1\to n-2}}(-1)^{j+ k_{1\to n-1}-j_{1\to n-2}}\tbinom{k_1}{j_1}\cdots\tbinom{k_{n-2}}{j_{n-2}} \partial^{(j)}{a^1}_{(j_1)}{a^2}_{(j_2)}\cdots{a^{n}}_{(j+ k_{1\to n-1}-j_{1\to n-2})}a^{n-1} $,
	
	(C$4^\prime$) $\sum\limits_{i=1}^{n}(-1)^{n-i}{a^1}_{(k_1)}\cdots{a^{i-1}}_{(k_{i-1})}{a^{i+1}}_{(k_{i+1})}{\cdots}_{(k_{n-1})}{a^n}_{(k_n)}({a^i}_{(k_i)}{b^2}_{(k_{n+1})}{\cdots}_{(k_{2n-2})}b^n)$
	
	\qquad$=\sum\limits_{j_{1\to n-1}}\tbinom{k_1}{j_1}\cdots\tbinom{k_{n-1}}{j_{n-1}}({a^1}_{(j_1)}{\cdots}_{(j_{n-1})}a^n)_{(+k_{1\to n}-j_{1\to n-1})}{b^2}_{(k_{n+1})}{\cdots}_{(k_{2n-2})}b^n$.
	
	\begin{example}
		Let $(\mathfrak{g},[\cdot,\cdots,\cdot])$ be an $n$-Lie algebra. Then $Cur\ \mathfrak{g}=\mathbb{C}[\partial]\otimes\mathfrak{g}$ carries an $n$-Lie conformal algebra structure given by $${g^1}_{(0)}{\cdots}_{(0)}g^n=[g^1,\cdots,g^n],\ {g^1}_{(k_1)}{\cdots}_{(k_{n-1})}g^n=0\mbox{ if }k_1+\cdots+k_{n-1}>0,$$ for $g^1,g^2,\cdots,g^n\in \mathfrak{g}$. The corresponding $\{\lambda_{1\to n-1}\}$-bracket is $$[{g^1}_{\lambda_1}{\cdots}_{\lambda_{n-1}}g^n]=[g^1,\cdots,g^n].$$ $Cur\ \mathfrak{g}$ is called the \textbf{current $n$-Lie conformal algebra} associated to $\mathfrak{g}$. It is not difficult to check that $Cur\ \mathfrak{g}$ is simple if and only if $\mathfrak{g}$ is simple and the isomorphism of $n$-Lie algebras corresponds to the isomorphism of associated current $n$-Lie conformal algebras. Since there exists only one finite-dimensional  simple $ n $-Lie algebra up to isomorphism \cite{Li}, finite simple current $n$-Lie conformal algebra is unique up to isomorphism.
	\end{example}

There is an $ n $-Lie algebra structure contained in the distribution notion of an $n$-Lie conformal algebra.
	\begin{lemma}\label{l1}
		Let $(R,\{{\bullet}_{(k_1)}{\bullet}_{(k_2)}{\cdots}_{(k_{n-1})}{\bullet}\})$ be an $n$-Lie conformal algebra. Then $ {\bullet}_{(0)}{\bullet}_{(0)}{\cdots}_{(0)}{\bullet} $ is an $n$-Lie bracket of $R$, and with respect to the $0$-th product $\partial R$ is an ideal of $R$ so that $R/\partial R$ is an $n$-Lie algebra.
	\end{lemma}
	\begin{proof}
		Denote $[a^1,a^2,\cdots,a^n]:={a^1}_{(0)}{a^2}_{(0)}{\cdots}_{(0)}a^n$ for all $a^1,a^2,\cdots,a^n\in R$. The skew-symmetry and Filippov identity can be easily obtained from (C$3^\prime$) and (C$4^\prime$). (C$2^\prime$) gives $$[a^1,\cdots,\partial a^i,\cdots,a^n]=0\mbox{ }if\mbox{ } i=1,2,\cdots,n-1,\mbox{ }and\mbox{ } [a^1,a^2,\cdots,\partial a^n]\in \partial R.$$ Therefore, $\partial R$ is an ideal of $R$ and $R/\partial R$ is an $n$-Lie algebra.
	\end{proof}

Kac et. al. classified Lie conformal algebras of rank one in \cite{ak} and showed Virasoro algebra is the only one non-commutative Lie conformal algebras of rank one. Interestingly, we find a different conclusion in $n$-Lie conformal algebra when $n\ge3$.
	\begin{proposition}\label{p2}
Let $R:=\mathbb{C}[\partial]e$ be an $n$-Lie conformal algebra of rank one for $n\geq 3$.  Then  $R$ is commutative.
	\end{proposition}
\begin{proof}
	Set $[e_{\lambda_1}e_{\lambda_2}{\cdots}_{\lambda_{n-1}}e]=f(\partial,\lambda_1,\cdots,\lambda_{n-1})e$. From (C2) and (C4), we have 
	\begin{align}\label{1}
	&f(-\lambda_1-\cdots-\lambda_n,\lambda_1,\cdots,\lambda_{n-1})f(\partial,\lambda_1+\cdots+\lambda_n,\lambda_{n+1},\cdots,\lambda_{2n-2})\notag\\
	=&\sum_{i=1}^{n}(-1)^{n-i}f(\partial+\sum_{\substack{j=1\\j\ne i}}^{n}\lambda_j,\lambda_i,\lambda_{n+1},\cdots,\lambda_{2n-2})f(\partial,\lambda_1,\cdots,\hat{\lambda_i},\cdots,\lambda_n).
	\end{align}
	Besides, (C3) implies
	$$f(\partial,\lambda_1,\cdots,\lambda_{n-1})=-f(\partial,\lambda_1,\cdots,\lambda_{j+1},\lambda_j,\cdots,\lambda_{n-1}),$$
	for $j=1,2,\cdots,n-2$.
	Hence if $\lambda_j=\lambda_k$ for some $j\ne k$ or $\lambda_j=0$ for some $j$, then we have $f(\partial,\lambda_1,\cdots,\lambda_{n-1})=0$.
	
	Plugging $\lambda_2=\lambda_{n+1},\lambda_3=\lambda_{n+2},\cdots,\lambda_{n-1}=\lambda_{2n-2}$ into (\ref{1}), we can obtain that 
	\begin{align}\label{2}
	&f(-\lambda_1-\cdots-\lambda_n,\lambda_1,\cdots,\lambda_{n-1})f(\partial,\lambda_1+\cdots+\lambda_n,\lambda_{n+1},\cdots,\lambda_{2n-2})\notag\\
	=&(-1)^{n-1}f(\partial+\sum_{j=2}^{n}\lambda_j,\lambda_1,\lambda_2,\cdots,\lambda_{n-1})f(\partial,\lambda_2,\cdots,\lambda_n).
	\end{align}
	 Substituting $\lambda_1+\cdots+\lambda_n=0$ into (\ref{2}), we obtain$$f(\partial+\sum_{j=2}^{n}\lambda_j,-\sum_{j=2}^{n}\lambda_j,\lambda_2,\cdots,\lambda_{n-1})f(\partial,\lambda_2,\cdots,\lambda_n)=0,$$ which implies either $f(\partial+\sum_{j=2}^{n}\lambda_j,-\sum_{j=2}^{n}\lambda_j,\lambda_2,\cdots,\lambda_{n-1})=0$ or $f(\partial,\lambda_2,\cdots,\lambda_n)=0$. In both cases, it can be inferred that $f(x_1,\cdots,x_n)=0$.
	 
	 Therefore, we have $[e_{\lambda_1}e_{\lambda_2}{\cdots}_{\lambda_{n-1}}e]=0$.
	 \end{proof}

However, one can easily find a non-commutative $n$-Lie conformal algebra free of rank two.
	\begin{proposition}
Let $R:=\mathbb{C}[\partial]e_1\oplus\mathbb{C}[\partial]e_2$ be an $n$-Lie conformal algebra free of rank two. If $I:=\mathbb{C}[\partial]e_2$ is a conformal ideal of $R$ and $n\ge3$, $R$ is determined by  $[{e_1}_{\lambda_1}{\cdots}_{\lambda_{n-2}}{e_1}_{\lambda_{n-1}}e_2]$ and $[{e_1}_{\lambda_1}{\cdots}_{\lambda_{n-2}}e_1]$.
		Furthermore, the brackets of $R$ are given by 
		
		(i) either \begin{align}
		[{e_{i_1}} _{\lambda_1}{e_{i_2}}_{\lambda_2}{\cdots}_{\lambda_{n-1}}e_{i_n}]=\begin{cases}
		0,&\mbox{if }i_1\cdots i_n\ge 2,\\
		\prod_{1\le i<j\le n-1}(\lambda_i-\lambda_j)g,&\mbox{if }i_1\cdots i_n=1,
		\end{cases}
		\end{align}
		where $i_1,\cdots,i_n\in\{1,2\}$, $g$ is a polynomial in $\mathcal{SP}_{n-1}[\partial]$ and $\mathcal{SP}_{n-1}$ is the set of all symmetric polynomials with $n-1$ variables.
		
		(ii) or \begin{align}
		[{e_{i_1}} _{\lambda_1}{e_{i_2}}_{\lambda_2}{\cdots}_{\lambda_{n-1}}e_{i_n}]=\begin{cases}
		0,&\mbox{if }i_1\cdots i_n\ge 4,\\
		h,&\mbox{if }i_1\cdots i_n=2,\\
		0,&\mbox{if }i_1\cdots i_n= 1,
		\end{cases}
		\end{align}
		where $i_1\le i_2\le\cdots\le i_n\in\{1,2\}$, and $h \in\mathcal{SP}_{n-1}$ satisfying
		\begin{align}\label{19}
		\sum_{i=1}^{n}(-1)^{n-i}h(\lambda_i,\lambda_{n+1},\cdots,\lambda_{2n-2})h(\lambda_1,\cdots,\hat{\lambda_i},\cdots,\lambda_n)=0. 
		\end{align}.
	\end{proposition}
\begin{proof}
	From Proposition \ref{p2}, $I$ and $R/I$ are commutative $n$-Lie conformal algebras, which implies 
	\begin{eqnarray*}
	&[{e_2} _{\lambda_1}{e_2}_{\lambda_2}{\cdots}_{\lambda_{n-1}}e_2]=0,\\
	&[{e_{i_1}} _{\lambda_1}{e_{i_2}}_{\lambda_2}{\cdots}_{\lambda_{n-1}}e_{i_n}]\in I[\lambda_1,\cdots,\lambda_{n-1}],\mbox{ for }i_1,\cdots,i_n\in\{1,2\}.
	\end{eqnarray*} 
	
	\textbf{Step 1:} Denote
	\begin{eqnarray*}
	&[{e_2}_{\lambda_1}{\cdots}_{\lambda_{n-2}}{e_2}_{\lambda_{n-1}}e_1]=f_1(\partial,\lambda_1,\cdots,\lambda_{n-1})e_2,\\
	 &[{e_2}_{\lambda_1}{\cdots}_{\lambda_{n-2}}{e_2}_{\lambda_{n-2}}{e_1}_{\lambda_{n-1}}e_1]=f_2(\partial,\lambda_1,\cdots,\lambda_{n-1})e_2,\\
	 &\vdots\\
	 &[{e_2}_{\lambda_1}{e_2}_{\lambda_2}{e_1}_{\lambda_3}{\cdots}_{\lambda_{n-1}}e_1]=f_{n-2}(\partial,\lambda_1,\cdots,\lambda_{n-1})e_2.
	\end{eqnarray*} We claim $f_1=f_2=\cdots=f_{n-2}=0$.
	
	 We first show $f_1=0$. It is easy to see that $f_1(\partial,\lambda_1,\cdots,\lambda_{n-1})$ is skew-symmetric with respect to $\lambda_i$ by (C3). From (C4),
	\begin{align*}
	&[[{e_2}_{\lambda_1}{\cdots}_{\lambda_{n-2}}{e_2}_{\lambda_{n-1}}e_1]_{\lambda_1+\cdots+\lambda_n}{e_2}_{\lambda_{n+1}}{\cdots}_{\lambda_{2n-3}}{e_2}_{\lambda_{2n-2}}e_1]\notag\\
	=&\sum_{i=1}^{n-1}(-1)^{n-i}[{e_2}_{\lambda_1}\cdots\hat{e_2}_{\lambda_i}\cdots{e_2}_{\lambda_{n-1}}{e_1}_{\lambda_n}[{e_2}_{\lambda_i}{e_2}_{\lambda_{n+1}}{\cdots}_{\lambda_{2n-3}}{e_2}_{\lambda_{2n-2}}e_1]],
	\end{align*}
	which is equivalent to\begin{align}\label{4}
	&f_1(-\lambda_1-\cdots-\lambda_n,\lambda_1,\cdots,\lambda_{n-1})f_1(\partial,\lambda_1+\cdots+\lambda_n,\lambda_{n+1},\cdots,\lambda_{2n-2})\notag\\
	=&\sum_{i=1}^{n-1}(-1)^{n+1-i}f_1(\partial+\sum_{\substack{j=1\\j\ne i}}^n\lambda_j,\lambda_i,\lambda_{n+1},\cdots,\lambda_{2n-2})f_1(\partial,\lambda_1,\cdots,\hat{\lambda_i},\cdots,\lambda_{n-1},-\partial-\sum_{\substack{j=1\\j\ne i}}^n\lambda_j).
	\end{align}
	Similar to Proposition \ref{p2}, we set$\lambda_2=\lambda_{n+1},\lambda_3=\lambda_{n+2},\cdots,\lambda_{n-1}=\lambda_{2n-2},\lambda_1+\cdots+\lambda_n=0$, then we can deduce that $f_1=0$ from the simplified formula of (\ref{4})  $$f_1(\partial+\sum_{j=2}^n\lambda_j,-(\sum_{j=2}^n\lambda_j),\lambda_2,\cdots,\lambda_{n-1})f_1(\partial,\lambda_2,\cdots,\lambda_{n-1},-\partial-(\sum_{\substack{j=1\\j\ne i}}^n\lambda_j))=0.$$ 
We further consider \begin{align*}
	&[[{e_2}_{\lambda_1}{\cdots}_{\lambda_{n-3}}{e_2}_{\lambda_{n-2}}{e_1}_{\lambda_{n-1}}e_1]_{\lambda_1+\cdots+\lambda_n}{e_2}_{\lambda_{n+1}}{\cdots}_{\lambda_{2n-4}}{e_2}_{\lambda_{2n-3}}{e_1}_{\lambda_{2n-2}}e_1]\notag\\
	=&\sum_{i=1}^{n-2}(-1)^{n-i}[{e_2}_{\lambda_1}\cdots\hat{e_2}_{\lambda_i}\cdots{e_2}_{\lambda_{n-2}}{e_1}_{\lambda_{n-1}}{e_1}_{\lambda_n}[{e_2}_{\lambda_i}{e_2}_{\lambda_{n+1}}{\cdots}_{\lambda_{2n-4}}{e_2}_{\lambda_{2n-3}}{e_1}_{\lambda_{2n-2}}e_1]],
	\end{align*}which implies \begin{align}
	&f_2(-\lambda_1-\cdots-\lambda_n,\lambda_1,\cdots,\lambda_{n-1})f_2(\partial,\lambda_1+\cdots+\lambda_n,\lambda_{n+1},\cdots,\lambda_{2n-2})\notag\\
	=&\sum_{i=1}^{n-2}(-1)^{n-i}f_2(\partial+\sum_{\substack{j=1\\j\ne i}}^n\lambda_j,\lambda_i,\lambda_{n+1},\cdots,\lambda_{2n-2})f_2(\partial,\lambda_1,\cdots,\hat{\lambda_i},\cdots,\lambda_{n-2},-\partial-\sum_{\substack{j=1\\j\ne i}}^n\lambda_j,\lambda_{n-1}).
	\end{align}
	Based on the skew-symmetry of $f_2$, substituting in $\lambda_2=\lambda_{n+1},\lambda_3=\lambda_{n+2},\cdots,\lambda_{n-2}=\lambda_{2n-3},\lambda_1+\cdots+\lambda_n=0$ yields $$f_2(\partial+\sum_{j=2}^n\lambda_j,-\sum_{j=2}^n\lambda_j,\lambda_2,\cdots,\lambda_{n-2},\lambda_{2n-2})f_2(\partial,\lambda_2,\cdots,\lambda_{n-2},-\partial-\sum_{j=2}^n\lambda_j,\lambda_{n-1})=0.$$ Hence $f_2=0$.
Similarly, we can obtain $f_1=\cdots=f_{n-2}=0$ immediately.
	
	\textbf{Step 2:} Let  $$[{e_1}_{\lambda_1}{\cdots}_{\lambda_{n-2}}{e_1}_{\lambda_{n-1}}e_2]:=f_{n-1}(\partial,\lambda_1,\cdots,\lambda_{n-1})e_2,$$ $$[{e_1}_{\lambda_1}{\cdots}_{\lambda_{n-2}}e_1]=f_n(\partial,\lambda_1,\cdots,\lambda_{n-1})e_2.$$ 
Plugging $\lambda_1=\lambda_n$ into 
	\begin{align*}
	&[[{e_1}_{\lambda_1}{\cdots}_{\lambda_{n-1}}e_1]_{\lambda_1+\cdots+\lambda_n}{e_1}_{\lambda_{n+1}}{\cdots}_{\lambda_{2n-2}}e_1]\\
	=&\sum_{i=1}^{n}(-1)^{n-i}[{e_1} _{\lambda_1}\cdots\hat{ e_1}_{\lambda_i}{\cdots}_{\lambda_{n-1}}{e_1}_{\lambda_n}[{ e_1}_{\lambda_i}{e_1}_{\lambda_{n+1}}{\cdots}_{\lambda_{2n-2}}e_1]],
	\end{align*}
	we have $$f_n(-\sum_{j=1}^{n-1}\lambda_j-2\lambda_1,\lambda_1,\cdots,\lambda_{n-1})f_{n-1}(\partial,\lambda_{n+1},\cdots,\lambda_{2n-2},-\partial-\lambda_i-\sum_{j=1}^{n-2}\lambda_{n+j})=0.$$ Thus either $f_n=0$ or $f_{n-1}=0$.
	
	 (i) If $f_{n-1}=0$, no matter what $f_n$ is, the Filippov identity always holds. The conformal skew-symmetry of $R$ yields the skew-symmetry of $f_n$ with respect to $\lambda_i$. So there is a one to one correspondence between $f_n$ and a polynomial $g$ in $\mathcal{SP}_{n-1}[\partial]$ by $$f_n=\prod_{1\le i<j\le n-1}(\lambda_i-\lambda_j)g.$$
	 
	 (ii) If $f_n=0$. It can be deduced from (C4) that 
	 \begin{align*}
	 &[[{e_1}_{\lambda_1}{\cdots}_{\lambda_{n-1}}e_1]_{\lambda_1+\cdots+\lambda_n}{e_1}_{\lambda_{n+1}}{\cdots}_{\lambda_{2n-3}}{e_1}_{\lambda_{2n-2}}e_2]\\
	 =&\sum_{i=1}^{n}(-1)^{n-i}[{e_1} _{\lambda_1}\cdots\hat{ e_1}_{\lambda_i}{\cdots}_{\lambda_{n-1}}{e_1}_{\lambda_n}[{ e_1}_{\lambda_i}{e_1}_{\lambda_{n+1}}{\cdots}_{\lambda_{2n-3}}{e_1}_{\lambda_{2n-2}}e_2]].
	 \end{align*}
	 Substituting in $f_{n-1}$, we obtain \begin{align}\label{5}
	 \sum_{i=1}^{n}(-1)^{n-i}f_{n-1}(\partial+\sum_{\substack{j=1\\j\ne i}}^n\lambda_j,\lambda_i,\lambda_{n+1},\cdots,\lambda_{2n-2})f_{n-1}(\partial,\lambda_1,\cdots,\hat{\lambda_i},\cdots,\lambda_n)=0.
	 \end{align}
	 Set $f_{n-1}(\partial,\lambda_1,\cdots,\lambda_{n-1})=\sum_{i=0}^{m}\limits g_i(\lambda_1,\cdots,\lambda_{n-1})\partial^i$ with $g_m\ne0$ and assume $m\ge1$. Computing the coefficients of $\partial^{2m-1}$, we have
	 \begin{align*}
	 &\sum_{i=1}^{n}(-1)^{n-i}g_m(\lambda_i,\lambda_{n+1},\cdots,\lambda_{2n-2})g_{m-1}(\lambda_1,\cdots,\hat{\lambda_i},\cdots,\lambda_n)\\
	 &+\sum_{i=1}^{n}(-1)^{n-i}g_{m-1}(\lambda_i,\lambda_{n+1},\cdots,\lambda_{2n-2})g_m(\lambda_1,\cdots,\hat{\lambda_i},\cdots,\lambda_n)\\
	 &+\sum_{i=1}^{n}(-1)^{n-i}m(\sum_{\substack{j=1\\j\ne i}}^n\lambda_j)g_m(\lambda_i,\lambda_{n+1},\cdots,\lambda_{2n-2})g_m(\lambda_1,\cdots,\hat{\lambda_i},\cdots,\lambda_n)=0.
	 \end{align*}
	 Since $g_i$ is skew-symmetric, putting in $\lambda_2=\lambda_{n+1},\cdots,\lambda_{n-1}=\lambda_{2n-2}$ yields 
	 $$m(\lambda_n-\lambda_1)g_m(\lambda_1,\cdots,\lambda_{n-1})g_m(\lambda_2,\cdots,\lambda_n)=0$$ so that $g_m=0$. This is a contradiction. Hence $f_{n-1}$ is a constant with respect to $\partial$. Then (\ref{5}) turns into
	 \begin{align}\label{6}
	 \sum_{i=1}^{n}(-1)^{n-i}f_{n-1}(\lambda_i,\lambda_{n+1},\cdots,\lambda_{2n-2})f_{n-1}(\lambda_1,\cdots,\hat{\lambda_i},\cdots,\lambda_n)=0. 
	 \end{align}
	 As mentioned above, it can be easily proved that $R$ is an $n$-Lie conformal algebra if and only if $f_{n-1}$ is a skew-symmetric polynomial with $n-1$ variables satisfying (\ref{6}).
\end{proof}
\begin{remark}
	Suppose that $n=3$ and $h(x,y)=\sum\limits_{i,j}a_{ij}x^iy^j\in\mathcal{SP}_2$, then $h(x,y)$ is a solution of (\ref{19}) if and only if the $m$ order antisymmetric matrix $\{a_{ij}\}$ satisfies $a_{ij}a_{kl}-a_{ik}a_{jl}+a_{il}a_{jk}=0$ for all $i,j,k,l=1,2,\cdots,m$. It can be easily checked that every 3 order antisymmetric matrix is a solution, but not all 4 order antisymmetric matrices meet the condition.
\end{remark}
	\section{The associated infinite-dimensional $n$-Lie algebras}\label{s3}
	In this section, we introduce the notion of associated infinite-dimensional $n$-Lie algebras $\mathscr{L}ie_p\mbox{ }R$ of an $n$-Lie conformal algebra $R$. Based on the structure of $\mathscr{L}ie_p\mbox{ }R$, we can define a linearly compact topology on the subalgebra $(\mathscr{L}ie_p\mbox{ }R)_\_$. We show torsionless finite $n$-Lie conformal algebra $R$ is determined by the linearly compact $n$-Lie algebra $(\mathscr{L}ie_p\mbox{ }R)_\_$, which is useful for further discussion on simple finite $n$-Lie conformal algebras.
	
	$\mathscr{L}ie_p\mbox{ }R$ are defined in the following way.
\begin{lemma}\label{p1}
	Suppose $R$ is an $n$-Lie conformal algebra, $\tilde{R}=R[t^{\pm1}]$. Then $\mathbb{C}[\tilde{\partial}]$-module $\tilde{R}$ is an $n$-Lie conformal algebra with the conformal products defined by
	\begin{align}\label{3}
	(a^1\otimes f_1)_{(k_1)}&{\cdots}_{(k_{n-1})}(a^n\otimes f_n)\notag\\
	=&\sum_{j_{1\to n-1}\in\mathbb{N}_+}({a^1}_{(k_1+j_1)}{\cdots}_{(k_{n-1}+j_{n-1})}a^n)\otimes((\partial_t^{(j_1)}f_1)\cdots(\partial_t^{(j_{n-1})}f_{n-1})f_n),
	\end{align} for all $a^1,\cdots,a^n\in R,f_1,\cdots,f_n\in \mathbb{C}[t^{\pm1}],k_1,\cdots,k_{n-1}\in\mathbb{N}_+$.
\end{lemma}
\begin{proof}
For convenience of calculations, we convert (\ref{3}) into the corresponding $\{\lambda_1,\cdots,\lambda_{n-1}\}$-brackets.
\begin{align}
&[{a^1\otimes f_1} _{\lambda_1}{\cdots}_{\lambda_{n-1}}a^n\otimes f_n]\notag\\
=&\sum_{k_{1\to n-1}\in\mathbb{N}_+}{\lambda_1}^{(k_1)}\cdots{\lambda_{n-1}}^{(k_{n-1})}(a^1\otimes f_1)_{(k_1)}{\cdots}_{(k_{n-1})}(a^n\otimes f_n)\notag\\
=&\sum_{\substack{k_{1\to n-1}\\j_{1\to n-1}}}{\lambda_1}^{(k_1)}\cdots{\lambda_{n-1}}^{(k_{n-1})}({a^1}_{(k_1+j_1)}{\cdots}_{(k_{n-1}+j_{n-1})}a^n)\otimes((\partial_t^{(j_1)}f_1)\cdots(\partial_t^{(j_{n-1})}f_{n-1})f_n)\notag\\
=&\sum_{\substack{k_{1\to n-1}\\j_{1\to n-1}}}{\lambda_1}^{(k_1)}\cdots{\lambda_{n-1}}^{(k_{n-1})}{\partial_{t_1}}^{(j_1)}\cdots{\partial_{t_{n-1}}}^{(j_{n-1})}({a^1}_{(k_1+j_1)}{\cdots}_{(k_{n-1}+j_{n-1})}a^n)\notag\\
&\qquad\qquad\qquad\qquad\qquad\qquad\qquad\qquad\qquad\qquad\otimes(f_1(t_1)\cdots f_{n-1}(t_{n-1})f_n(t))|_{t_1=t,\cdots,t_{n-1}=t}\notag\\
=&\sum_{\substack{k_1+j_1\\\cdots\\k_{n-1}+j_{n-1}}}{(\lambda_1+\partial_{t_1})}^{(k_1+j_1)}\cdots{(\lambda_{n-1}+\partial_{t_{n-1}})}^{(k_{n-1}+j_{n-1})}({a^1}_{(k_1+j_1)}{\cdots}_{(k_{n-1}+j_{n-1})}a^n)\notag\\
&\qquad\qquad\qquad\qquad\qquad\qquad\qquad\qquad\qquad\qquad\otimes(f_1(t_1)\cdots f_{n-1}(t_{n-1})f_n(t))|_{t_1=t,\cdots,t_{n-1}=t}\notag\\
=&[{a^1}_{\lambda_1+\partial_{t_1}}{\cdots}_{\lambda_{n-1}+\partial_{t_{n-1}}}a^n]\otimes (f_1(t_1)\cdots f_{n-1}(t_{n-1})f_n(t))|_{t_1=t,\cdots,t_{n-1}=t}
\end{align} 
Then we only need to verify (C1)-(C4) are satisfied. Let us check (C3) as an example.

a). For $i=1,2,\cdots,n-2$,
\begin{align*}
&[{a^1\otimes f_1} _{\lambda_1}\cdots {a^{i+1}\otimes f_{i+1}} _{\lambda_{i+1}}{a^i\otimes f_i} _{\lambda_i}\cdots a^n\otimes f_n]\\
=&[{a^1}_{\lambda_1+\partial_{t_1}}\cdots{a^{i+1}}_{\lambda_{i+1}+\partial_{t_{i+1}}} {a^i}_{\lambda_i+\partial_{t_i}}\cdots a^n]\otimes (f_1(t_1)\cdots f_{n-1}(t_{n-1})f_n(t))|_{t_1=t,\cdots,t_{n-1}=t}\\
=&-[{a^1}_{\lambda_1+\partial_{t_1}}\cdots {a^i}_{\lambda_i+\partial_{t_i}}{a^{i+1}}_{\lambda_{i+1}+\partial_{t_{i+1}}}\cdots a^n]\otimes (f_1(t_1)\cdots f_{n-1}(t_{n-1})f_n(t))|_{t_1=t,\cdots,t_{n-1}=t}\\
=&-[{a^1\otimes f_1} _{\lambda_1}{\cdots}_{\lambda_{n-1}}a^n\otimes f_n].
\end{align*}

b).\begin{align*}
&[{a^1\otimes f_1} _{\lambda_1}\cdots {a^{n-1}\otimes f_{n-1}} _{\lambda_{n-1}}a^n\otimes f_n]\\
=&[{a^1}_{\lambda_1+\partial_{t_1}}\cdots{a^{n-1}}_{\lambda_{n-1}+\partial_{t_{n-1}}}a^n]\otimes (f_1(t_1)\cdots f_{n-1}(t_{n-1})f_n(t))|_{t_1=t,\cdots,t_{n-1}=t}\\
=&-[{a^1}_{\lambda_1+\partial_{t_1}}\cdots {a^n} _{-\partial-\lambda_{1\to n-1}-\partial_{t_{1\to n}}+\partial_{t_n}}a^{n-1}]\otimes(f_1(t_1)\cdots f_{n-1}(t_{n-1})f_n(t_n))|_{t_1=t,\cdots,t_n=t}\\
=&-[{a^1}_{\lambda_1+\partial_{t_1}}\cdots {a^n} _{-\tilde{\partial}-\lambda_{1\to n-1}+\partial_{t_n}}a^{n-1}]\otimes(f_1(t_1)\cdots f_{n-1}(t_{n-1})f_n(t_n))|_{t_1=t,\cdots,t_n=t}\\
=&-[{a^1\otimes f_1} _{\lambda_1}\cdots {a^n\otimes f_n} _{-\tilde{\partial}-\lambda_{1\to n-1}}a^{n-1}\otimes f_{n-1}]
\end{align*}
\end{proof}
\begin{corollary}\label{c1}
(i) For each $p\ge1,p\in\mathbb{N}_+$, $\tilde{R}_p:=R[t_1^{\pm1},\cdots,t_p^{\pm1}]$ is an $n$-Lie conformal algebra with $\tilde{\partial}_p:=\partial\otimes1\cdots1+1\otimes\partial_{t_1}\otimes1\cdots1+\cdots+1\otimes\cdots\otimes1\otimes\partial_{t_p}$ and the conformal products\begin{align}
&(a^1\otimes f_1)_{(k_1)}{\cdots}_{(k_{n-1})}(a^n\otimes f_n)\notag\\
=&\sum_{j^{1\to n-1}_{1\to p}}({a^1}_{(k_1+j^1_{1\to p})}{\cdots}_{(k_{n-1}+j_{1\to p}^{n-1})}a^n)\otimes((\partial_{t_1}^{(j^1_1)}\cdots\partial_{t_p}^{(j_p^1)}f_1)\cdots(\partial_{t_1}^{(j_1^{n-1})}\cdots\partial_{t_p}^{(j^{n-1}_p)}f_{n-1})f_n).
\end{align}

(ii) Denote $\mathscr{L}ie_p\mbox{ }R:=\tilde{R}_p/\tilde{\partial}_p\tilde{R}_p$ and $a_{m_{1\to p}}=a_{m_1,\cdots,m_p}:=a\otimes t_1^{m_1}\cdots t_p^{m_p}\in\mathscr{L}ie_p\mbox{ }R$. The bracket
\begin{align}\label{7}
&[a^1_{m_{1\to p}^1},\cdots,a^n_{m_{1\to p}^n}]\notag\\
=&\sum_{j^{1\to n-1}_{1\to p}}\tbinom{m_1^1}{j_1^1}\cdots\tbinom{m_p^{n-1}}{j_p^{n-1}}({a^1}_{(j^1_{1\to p})}{\cdots}_{(j_{1\to p}^{n-1})}a^n)_{m_1^n+m_1^{1\to n-1}-j_1^{1\to n-1},\cdots,m_p^n+m_p^{1\to n-1}-j_p^{1\to n-1}}\notag\\
\end{align} is a well-defined $n$-Lie bracket on $\mathscr{L}ie_p\mbox{ }R$. Besides, $\mathscr{L}ie_p\mbox{ }R$ is a $\mathbb{C}[\partial]$-module via \begin{eqnarray}\label{8}
\partial(a_{m_1,\cdots,m_p})=-(\sum_{i=0}^{p}m_ia_{m_1,\cdots,m_i-1,\cdots m_p}).
\end{eqnarray} 

(iii) The induced $\partial_{t_i}$ are derivations of $\mathscr{L}ie_p\mbox{ }R$, given by
\begin{eqnarray}\label{10}
\partial_{t_i}(a_{m_1,\cdots,m_p})=m_ia_{m_1,\cdots,m_i-1,\cdots m_p}.
\end{eqnarray}  Thus $\partial$ in (ii) is also a derivation of $\mathscr{L}ie_p\mbox{ }R$.
\end{corollary}
\begin{proof}
(i)	Based on the fact that $\tilde{R}_{p+1}=\tilde{R}_p[t_{p+1}^{\pm1}]$, the result can be inductively obtained from Lemma \ref{p1}.

(ii) The statement is obviously true by (i) and Lemma \ref{l1}.

(iii) Since $\partial_{t_i}$ commutes with $\tilde{\partial}_p$, (\ref{10}) is well-defined. The statement can be proved by verifying $$\partial_{t_i}[a^1_{m_{1\to p}^1},\cdots,a^n_{m_{1\to p}^n}]=\sum_{j=1}^{n}[a^1_{m_{1\to p}^1},\cdots,\partial_{t_i}(a^j_{m_{1\to p}^j}),\cdots,a^n_{m_{1\to p}^n}].$$ Take $\partial_{t_1}$ as an example.
\begin{align*}
&\sum_{j=1}^{n}[a^1_{m_1^1,\cdots,m_p^1},\cdots,\partial_{t_1}(a^j_{m_1^j,\cdots,m_p^j}),\cdots,a^n_{m_1^n,\cdots,m_p^n}]\\
=&\sum_{j=1}^{n}m_1^j[a^1_{m_1^1,\cdots,m_p^1},\cdots,a^j_{m_1^j-1,m_2^j\cdots,m_p^j},\cdots,a^n_{m_1^n,\cdots,m_p^n}]\\
=&\sum_{j=1}^{n}\sum_{j_1^1,\cdots,j^{n-1}_p}m_1^j\tbinom{m_1^1}{j_1^1}\cdots\tbinom{m_1^j-1}{j_1^j}\cdots\tbinom{m_p^{n-1}}{j_p^{n-1}}\notag\\
&\qquad({a^1}_{(j^1_1+\cdots+j_p^1)}{\cdots}_{(j_1^{n-1}+\cdots+j_p^{n-1})}a^n)_{m_1^1+\cdots+m_1^n-j_1^1-\cdots-j_1^{n-1}-1,\cdots,m_p^1+\cdots+m_p^n-j_p^1-\cdots-j_p^{n-1}}\notag\\
=&\sum_{j_1^1,\cdots,j^{n-1}_p}\sum_{j=1}^{n}(m_1^j-j_1^j)\tbinom{m_1^1}{j_1^1}\cdots\tbinom{m_1^j}{j_1^j}\cdots\tbinom{m_p^{n-1}}{j_p^{n-1}}\notag\\
&\qquad({a^1}_{(j^1_1+\cdots+j_p^1)}{\cdots}_{(j_1^{n-1}+\cdots+j_p^{n-1})}a^n)_{m_1^1+\cdots+m_1^n-j_1^1-\cdots-j_1^{n-1}-1,\cdots,m_p^1+\cdots+m_p^n-j_p^1-\cdots-j_p^{n-1}}\notag\\
=&\sum_{j_1^1,\cdots,j^{n-1}_p}(m_1^1+\cdots+m_1^n-j_1^1-\cdots-j_1^{n-1})\tbinom{m_1^1}{j_1^1}\cdots\tbinom{m_1^j}{j_1^j}\cdots\tbinom{m_p^{n-1}}{j_p^{n-1}}\notag\\
&\qquad({a^1}_{(j^1_1+\cdots+j_p^1)}{\cdots}_{(j_1^{n-1}+\cdots+j_p^{n-1})}a^n)_{m_1^1+\cdots+m_1^n-j_1^1-\cdots-j_1^{n-1}-1,\cdots,m_p^1+\cdots+m_p^n-j_p^1-\cdots-j_p^{n-1}}\notag\\
=&\partial_{t_1}[a^1_{m_1^1,\cdots,m_p^1},\cdots,a^n_{m_1^n,\cdots,m_p^n}].
\end{align*}
Furthermore, $\partial=-(\partial_{t_1}+\cdots+\partial_{t_p})$ is also a derivation of $\mathscr{L}ie_p\mbox{ }R$.
\end{proof}

Then let us discuss the structure of $\mathscr{L}ie_p\mbox{ }R$.

\begin{proposition}\label{p3}
(i)	If $R$ is a free $\mathbb{C}[\partial]$-module, $\mathscr{L}ie_p\mbox{ }R\simeq V[t_1^{\pm1},\cdots,t_p^{\pm1}]$ for some $\mathbb{C}$-space $V$.

(ii) Let $R$ be a finite generated torsion $\mathbb{C}[\partial]$-module. Then $r_{m_1,\cdots,m_p}=0$ in $\mathscr{L}ie_p\mbox{ }R$ if $m_1,\cdots,m_p\ge0$. Particularly, $\mathscr{L}ie_1\mbox{ }R\simeq Rt_1^{-1}$.
	\end{proposition}
\begin{proof}
	(i) With the assumption, we can find a  $\mathbb{C}$-space $V$ spanned by $\mathbb{C}[\partial]$-basis of $R$ so that $R\simeq \mathbb{C}[\partial]\otimes V$. Then $\tilde{R}_p\simeq \mathbb{C}[\partial]\otimes V[t_1^{\pm1},\cdots,t_p^{\pm1}]$. For any $\sum\limits_{i=0}^{m}\partial^i\otimes v_i(t_1,\cdots,t_p) \in \mathbb{C}[\partial]\otimes V[t_1^{\pm1},\cdots,t_p^{\pm1}]$, the relation\begin{align*}
	&\sum_{i=0}^{m}\partial^i\otimes v_i(t_1,\cdots,t_p)=v_0+\sum_{i=1}^{m}\partial^i\otimes v_i\\
	=&v_0+\sum_{i=1}^{m}(\partial^{(i-1)+1}\otimes v_i+\partial^{i-1}\otimes(\partial_{t_1}+\cdots+\partial_{t_p})v_i)-(\sum_{i=1}^{m}\partial^{i-1}\otimes(\partial_{t_1}+\cdots+\partial_{t_p})v_i)\\
	=&v_0-(\sum_{i=1}^{m}\partial^{i-1}\otimes(\partial_{t_1}+\cdots+\partial_{t_p})v_i)+\sum_{i=1}^{m}\tilde{\partial}_p(\partial^{i-1}\otimes v_i)\\
	\equiv&v_0-(\sum_{i=1}^{m}\partial^{i-1}\otimes(\partial_{t_1}+\cdots+\partial_{t_p})v_i)\\
	\equiv&v_0-(\partial_{t_1}+\cdots+\partial_{t_p})v_1+(\sum_{i=2}^{m}\partial^{i-2}\otimes(\partial_{t_1}+\cdots+\partial_{t_p})^2v_i)\\
	&\cdots\\
	\equiv&\sum_{i=0}^{m}(-1)^i(\partial_{t_1}+\cdots+\partial_{t_p})^iv_i
	\end{align*}
	holds in $\mathscr{L}ie_p\mbox{ }R$. That is,  $\tilde{R}_p=V[t_1^{\pm1},\cdots,t_p^{\pm1}]+\tilde{\partial}_p\tilde{R}_p$.
	
	Next, we show the sum is direct. Choose a basis $\{e_i\}$ of $V$, which is also a $\mathbb{C}[\partial]$-basis of $R$. Then up to an isomorphism, we can denote an arbitrary element of $V[t_1^{\pm1},\cdots,t_p^{\pm1}]$ as $\sum\limits_{i=1}^{s}e_i\otimes f_i$ and an element of $\tilde{R}_p$ as $\sum\limits_{i=0}^{r}\sum\limits_{j=1}^{l}\partial_ie_j\otimes f_{ij}$. Thus for an element in the intersection of $V[t_1^{\pm1},\cdots,t_p^{\pm1}]$ and $\tilde{\partial}_p\tilde{R}_p$, we have 
	$$\tilde{\partial}_p(\sum\limits_{i=0}^{r}\sum\limits_{j=1}^{l}\partial_ie_j\otimes f_{ij})=\sum\limits_{i=1}^{s}e_i\otimes f_i,$$
	Since \begin{align}
	&\tilde{\partial}_p(\sum\limits_{i=0}^{r}\sum\limits_{j=1}^{l}\partial^ie_j\otimes f_{ij})=\sum\limits_{i=0}^{r}\sum\limits_{j=1}^{l}(\partial^{i+1}e_j\otimes f_{ij}+\partial^ie_j\otimes(\partial_{t_1}+\cdots+\partial_{t_p})f_{ij})\notag\\
	=&\sum\limits_{j=1}^{l}e_j\otimes(\partial_{t_1}+\cdots+\partial_{t_p})f_{0j}+\sum\limits_{i=1}^{r}\sum\limits_{j=1}^{l}(\partial^{i}e_j\otimes f_{i-1\ j}+(\partial_{t_1}+\cdots+\partial_{t_p})f_{ij})+\sum\limits_{j=1}^{l}\partial^{r+1}e_j\otimes f_{rj}\notag
	\end{align} and $\{\partial^ie_j\}$ are independent, we can find successively $f_{rj}=f_{r-1\ j}=\cdots=f_{0j}=0$ for all $j=1,2,\cdots,l$. So there has no non-zero element in the intersection of $V[t_1^{\pm1},\cdots,t_p^{\pm1}]$ and $\tilde{\partial}_p\tilde{R}_p$.
	
	Therefore, $\mathscr{L}ie_p\mbox{ }R\simeq V[t_1^{\pm1},\cdots,t_p^{\pm1}]$.
	
		(ii) From \cite{n}, $R=\oplus_{i=1}^{h}R_{\partial-c_i}$ where $R_{\partial-c_i}:=\{r\in R|(\partial-c_i)^jr=0\mbox{ for some }j \in \mathbb{N}_+\}$ and $c_i$ are fixed constants. Abbreviating $R_{\partial-c_i}[t_1^{\pm1},\cdots,t_p^{\pm1}]$ as $\tilde{R}_{\partial-c_i}$, we have $\mathscr{L}ie_p\mbox{ }R\simeq\oplus_{i=1}^{h}\tilde{R}_{\partial-c_i}/\tilde{\partial}_p\tilde{R}_{\partial-c_i}$. So without loss of generality, we assume that $R=R_{\partial-c}$ for a certain constant $c$. For any $\sum_{i_1,\cdots,i_p}\limits r^{i_1,\cdots,i_p}\otimes t_1^{i_1}\cdots t_p^{i_p}\in R[t_1,\cdots,t_p]$, there exists $\sum_{j_1,\cdots,j_p}\limits s^{j_1,\cdots,j_p}\otimes t_1^{j_1}\cdots t_p^{j_p}\in R[t_1,\cdots,t_p]$ such that $$\tilde{\partial}_p(\sum_{j_1,\cdots,j_p}s^{j_1,\cdots,j_p}\otimes t_1^{j_1}\cdots t_p^{j_p})=\sum_{i_1,\cdots,i_p}r^{i_1,\cdots,i_p}\otimes t_1^{i_1}\cdots t_p^{i_p}.$$ It can be proven by induction on $p$.

	 Let $p=1$. For $r\otimes t_1^m\in R[t_1^{\pm1}]$ where $(\partial-c)^jr=0$ for a fixed $j$, it can be computed that
	\begin{align}
	r\otimes t_1^m=\begin{cases}
	\tilde{\partial}_1(\sum\limits_{i=1}^{j}\frac{m!}{(m+i)!}(-\partial)^{i-1}r\otimes t_1^{m+i}),&\mbox{if } m\ge0, c=0,\\
	-\tilde{\partial}_1(\sum\limits_{i=0}^{m}\frac{m!}{i!}(-f(\partial))^{m+1-i}r\otimes t_1^i),&\mbox{if }m\ge0, c\ne0,\\
	-\tilde{\partial}_1(\sum\limits_{i=0}^{-m-1}\frac{(-m-1-i)!}{(-m-1)!}(-\partial)^{i-1}r\otimes t_1^{m+i})-\frac{1}{(-m-1)!}(-\partial)^{i-1}r\otimes t_1^{-1},&\mbox{if }m\le-1,
	\end{cases}
	\end{align}
	where $f(\partial)$ is a polynomial in $\mathbb{C}[\partial]$ such that $\partial f(\partial)r=r$. By $(\tfrac{\partial}{c}-1)^jr=0$, $f$ is well-defined. Besides, that $r\otimes t_1^{-1}\in\tilde{\partial}_1\tilde{R}_1$ connotes the existence of $\sum\limits_{i=-l}^{q}r_i\otimes t_1^i$ such that
	\begin{align*}
	r\otimes t_1^{-1}=&\tilde{\partial}_1(\sum_{i=-l}^{q}r_i\otimes t_1^i)\\
	=&-lr_{-l}\otimes t_1^{-l-1}+\sum_{i=-l}^{-2}(\partial r_i+(i+1)r_{i+1})\otimes t_1^i+\partial r_{-1}\otimes t_1^{-1}\\
	&\qquad\qquad+\sum_{i=0}^{q-1}(\partial r_i+(i+1)r_{i+1})\otimes t_1^i+\partial r_q\otimes t_1^q.
	\end{align*} 
	Contrasting the coefficients, we have $r_{-l}=0$ when $l$ is large enough. Then $r_{-l+1}=r_{-l+2}=\cdots=r_{-1}=0$ are deduced in succession. So $r=\partial r_{-1}=0$ and thus $\tilde{R}_1=Rt_1^{-1}\oplus \tilde{\partial}_1\tilde{R}_1$. Hence $\mathscr{L}ie_1\mbox{ }R\simeq Rt_1^{-1}$.
	
	Now assume that the statement is true for $p=k$. Let us consider $r\otimes t_1^{m_1}\cdots t_{k+1}^{m_{k+1}}$ where $m_1,\cdots,m_{k+1}\ge0$. Based on the assumption, there exists $f_{m_{k+1}}\in R[t_1,\cdots,t_k]$ such that $\tilde{\partial}_kf_{m_{k+1}}=r\otimes t_1^{m_1}\cdots t_k^{m_k}$. Similarly, $f_{m_{k+1}-1},f_{m_{k+1}-2},\cdots,f_0$ can be found by recurrence relations $\tilde{\partial}_kf_i=-(i+1)f_{i+1}$. Then we have $$r\otimes t_1^{m_1}\cdots t_{k+1}^{m_{k+1}}=\tilde{\partial}_{k+1}(\sum_{i=0}^{m_{k+1}}f_i\otimes t_{k+1}^i).$$
\end{proof}

From now on, let  $R$ be a finite $n$-Lie conformal algebra. By \cite{n}, there is a finite-dimensional space $V$ such that
$$R=(\mathbb{C}[\partial]\otimes V)\oplus Tor\mbox{ }R,$$
where $Tor\mbox{ }R:=\{r\in R|f(\partial)r=0 \mbox{ for some non-zero }f(\partial) \in \mathbb{C}[\partial]\}$.
Choose $\{v^i\}$ as a basis of $V$, and we can define a series of subspaces of $\mathscr{L}ie_p\mbox{ }R$
\begin{eqnarray}\label{11}
\{\mathscr{L}_p^{(m)}=span\{v^i_{m_1,\cdots,m_p}|m_1,\cdots,m_p\ge0,m_1+\cdots+m_p\ge m\}\}.
\end{eqnarray} 

Recall that a topological space is called a linearly compact space if it is a product of finite dimensional discrete spaces \cite{l,g}, and a linearly compact algebra is a linearly compact space equipped with a continuous structure map \cite{g,ck}. From the above definition, we obtain the following  four  results immediately.

(1) From (\ref{7}) and (\ref{8}), $\mathscr{L}_p^{(0)}$ is a $\partial$-stable subalgebra of $\mathscr{L}ie_p\mbox{ }R$. Furthermore, if $\{v^i\}$ are $\mathbb{C}[\partial]$-independent in $R$, $\mathscr{L}_p^{(0)}\simeq V[t_1,\cdots,t_p]$ by Proposition \ref{p3}.

(2) If $i<j$, $\mathscr{L}_p^{(i)}\supset \mathscr{L}_p^{(j)}$ and $\mathscr{L}_p^{(i)}/\mathscr{L}_p^{(j)}$ is of finite dimension.

(3) For any $i_1,\cdots,i_p\ge0$, $[\mathscr{L}_p^{(i_1)},\cdots,\mathscr{L}_p^{(i_p)}]\subset \mathscr{L}_p^{(i_1+\cdots+i_p-s)}$ for some $s$ by (\ref{7}).

(4)  $\mathscr{L}_p^{(0)}$ can be seen as a linearly compact $n$-Lie algebra with the topology induced by $\{\mathscr{L}_p^{(m)}\}$, as it is the topological product of $\mathscr{L}_p^{(i)}/\mathscr{L}_p^{(i+1)},i\ge0$ as a space and has a continuous $n$-Lie bracket as an algebra. 

It should be noted that although $\mathscr{L}_p^{(m)}$ depends on the choice of $\{v^i\}$, $\mathscr{L}_p^{(0)}$ and the induced topology on $\mathscr{L}_p^{(0)}$ is independent of that choice. So we can have the following definition.
\begin{definition}
	\begin{em}
	Denote $(\mathscr{L}ie_p\mbox{ }R)_\_:=\mathscr{L}_p^{(0)}=\{a_{m_1,\cdots,m_p}|a\in R,m_1,\cdots,m_p\ge0\}.$ Then the topological $n$-Lie algebra $(\mathscr{L}ie_p\mbox{ }R)_\_$ is called the \textit{annihilation algebra of level $p$} associated to $n$-Lie conformal algebra $R$. The completion $\widehat{(\mathscr{L}ie_p\mbox{ }R)_\_}$ of $(\mathscr{L}ie_p\mbox{ }R)_\_$ with respect to its topology is called the \textit{completed annihilation algebra of level $p$} of $R$.
	\end{em}
\end{definition}
\begin{example}
	For the current $n$-Lie conformal algebra $Cur\ \mathfrak{g}$ associated to an $n$-Lie algebra $\mathfrak{g}$, $(\mathscr{L}ie_p\mbox{ }Cur\ \mathfrak{g})_\_$ is nothing but the $n$-Lie algebra $\mathfrak{g}[t_1,\cdots,t_p]$ endowed with the topology induced by $\{\mathscr{L}_p^{(m)}=\sum\limits_{m_1+\cdots+m_p=m}t_1^{m_1}\cdots t_p^{m_p}\mathfrak{g}[t_1,\cdots,t_p]\}$.
\end{example} 
A finite $n$-Lie conformal algebra $R$ is called \textit{torsionless} if $Tor\ R=0$. To show the importance of $ (\mathscr{L}ie_p\mbox{ }R)_\_ $ to torsionless finite $R$, the following lemma is necessary.
\begin{lemma}
	For $a^1,\cdots,a^n\in R,m_1^1,\cdots,m_p^n\in\mathbb{N}$,
\begin{align}\label{9}
&({a^1}_{(\sum m^1_{1\to p})}{\cdots}_{(\sum m^{n-1}_{1\to p})}a^n)_{m_{1\to p}^n}=\sum_{j_{1\to p}^{1\to n-1}}(-1)^{\sum m_{1\to p}^{1\to n-1}-\sum j_{1\to p}^{1\to n-1}}\tbinom{m_1^1}{j_1^1}\cdots\tbinom{m_p^{n-1}}{j_p^{n-1}}\notag\\
&\qquad\qquad\qquad\qquad\qquad[a^1_{j_{1\to p}^1},\cdots,a^{n-1}_{j_{1\to p}^{n-1}},a^n_{m_1^n+m_1^{1\to n-1}-j_1^{1\to n-1},\cdots,m_p^n+m_p^{1\to n-1}-j_p^{1\to n-1}}].
\end{align}
\end{lemma}
\begin{proof}
	Let $a_{\lambda_1,\cdots,\lambda_p}:=\sum\limits_{m_1,\cdots,m_p\in\mathbb{N}_+}{\lambda_1}^{(m_1)}\cdots{\lambda_p}^{(m_p)} a_{m_1,\cdots,m_p}$, and we have
	\begin{align}
	&[a^1_{\lambda^1_1,\cdots,\lambda_p^1},\cdots,a^n_{\lambda^n_1,\cdots,\lambda_p^n}]\notag\\
	=&\sum_{m_1^1,\cdots,m_p^n}{\lambda^1_1}^{(m_1^1)}\cdots{\lambda^n_p}^{(m_p^n)}[a^1_{m_1^1,\cdots,m_p^1},\cdots,a^n_{m_1^n,\cdots,m_p^n}]\notag\\
	=&\sum_{m_1^1,\cdots,m_p^n,j_1^1,\cdots,j_p^{n-1}}\tbinom{m_1^1}{j_1^1}\cdots\tbinom{m_p^{n-1}}{j_p^{n-1}}{\lambda^1_1}^{(m_1^1)}\cdots{\lambda^n_p}^{(m_p^n)}\notag\\
	&\qquad({a^1}_{(j^1_1+\cdots+j_p^1)}{\cdots}_{(j_1^{n-1}+\cdots+j_p^{n-1})}a^n)_{m_1^1+\cdots+m_1^n-j_1^1-\cdots-j_1^{n-1},\cdots,m_p^1+\cdots+m_p^n-j_p^1-\cdots-j_p^{n-1}}\notag\\
	=&\sum_{{\tiny \shortstack{$ j^1_1+\cdots+j_p^1,\cdots,j_1^{n-1}+\cdots+j_p^{n-1}, $\\$ m_1^1+\cdots+m_1^n-j_1^1-\cdots-j_1^{n-1},\cdots,m_p^1+\cdots+m_p^n-j_p^1-\cdots-j_p^{n-1} $}}}(\lambda^1_1+\cdots+\lambda_p^1)^{(j^1_1+\cdots+j_p^1)}\cdots\notag\\
	&\qquad(\lambda^{n-1}_1+\cdots+\lambda_p^{n-1})^{(j^{n-1}_1+\cdots+j_p^{n-1})}(\lambda^1_1+\cdots+\lambda_1^n)^{(m_1^1+\cdots+m_1^n-j_1^1-\cdots-j_1^{n-1})}\cdots\notag\\
	&\qquad(\lambda^1_p+\cdots+\lambda_p^n)^{(m_p^1+\cdots+m_p^n-j_p^1-\cdots-j_p^{n-1})}\notag\\
	&\qquad({a^1}_{(j^1_1+\cdots+j_p^1)}{\cdots}_{(j_1^{n-1}+\cdots+j_p^{n-1})}a^n)_{m_1^1+\cdots+m_1^n-j_1^1-\cdots-j_1^{n-1},\cdots,m_p^1+\cdots+m_p^n-j_p^1-\cdots-j_p^{n-1}}\notag\\
	=&[{a^1}_{\lambda^1_1+\cdots+\lambda_p^1}{\cdots}_{\lambda^{n-1}_1+\cdots+\lambda_p^{n-1}}a^n]_{\lambda^1_1+\cdots+\lambda_1^n,\cdots,\lambda^1_p+\cdots+\lambda_p^n}
	\end{align}
	Replacing $\lambda^1_1+\cdots+\lambda_1^n,\cdots,\lambda^1_p+\cdots+\lambda_p^n$ by $\lambda_1^n,\cdots,\lambda_p^n$ yields $$[{a^1}_{\lambda^1_1+\cdots+\lambda_p^1}{\cdots}_{\lambda^{n-1}_1+\cdots+\lambda_p^{n-1}}a^n]_{\lambda_1^n,\cdots,\lambda_p^n}=[a^1_{\lambda^1_1,\cdots,\lambda_p^1},\cdots,a^n_{\lambda^n_1-\lambda^1_1-\cdots-\lambda^{n-1}_1,\cdots,\lambda_p^n-\lambda^1_p-\cdots-\lambda_p^{n-1}}].$$
	Then (\ref{9}) can be deduced from 
		\begin{align}
	&[{a^1}_{\lambda^1_1+\cdots+\lambda_p^1}{\cdots}_{\lambda^{n-1}_1+\cdots+\lambda_p^{n-1}}a^n]_{\lambda_1^n,\cdots,\lambda_p^n}\notag\\
	=&\sum_{m_1^1,\cdots,m_p^n}{\lambda^1_1}^{(m_1^1)}\cdots{\lambda^n_p}^{(m_p^n)}({a^1}_{(m^1_1+\cdots+m_p^1)}{\cdots}_{(m^{n-1}_1+\cdots+m_p^{n-1})}a^n)_{m_1^n,\cdots,m_p^n}\notag\\
	=&[a^1_{\lambda^1_1,\cdots,\lambda_p^1},\cdots,a^n_{\lambda^n_1-\lambda^1_1-\cdots-\lambda^{n-1}_1,\cdots,\lambda_p^n-\lambda^1_p-\cdots-\lambda_p^{n-1}}]\notag\\
	=&\sum_{k_1^1,\cdots,k_p^n}{\lambda^1_1}^{(k_1^1)}\cdots{\lambda^{n-1}_p}^{(k_p^{n-1})}(\lambda^n_1-\lambda^1_1-\cdots-\lambda^{n-1}_1)^{(k^n_1)}\cdots(\lambda_p^n-\lambda^1_p-\cdots-\lambda_p^{n-1})^{(k_p^n)}\notag\\
	&\qquad[a^1_{k_1^1,\cdots,k_p^1},\cdots,a^n_{k_1^n,\cdots,k_p^n}]\notag\\
	=&\sum_{m_1^1,\cdots,m_p^n,j_1^1,\cdots,j_p^{n-1}}{\lambda^1_1}^{(m_1^1)}\cdots{\lambda^n_p}^{(m_p^n)}(-1)^{m_1^1+\cdots+m_p^{n-1}-j_1^1-\cdots-j_p^{n-1}}\tbinom{m_1^1}{j_1^1}\cdots\tbinom{m_p^{n-1}}{j_p^{n-1}}\notag\\
	&\qquad[a^1_{j_1^1,\cdots,j_p^1},\cdots,a^{n-1}_{j_1^{n-1},\cdots,j_p^{n-1}},a^n_{m_1^1+\cdots+m_1^n-j_1^1-\cdots-j_1^{n-1},\cdots,m_p^1+\cdots+m_p^n-j_p^1-\cdots-j_p^{n-1}}].
	\end{align}
	
\end{proof}

\begin{corollary}\label{c2}
	(i) $a_{\lambda_1,\cdots,\lambda_p}=0$ if and only if $a\in Tor\mbox{ }R$.
	
	(ii) Let $R$ be  an $n$-Lie conformal algebra $R$ and a free $\mathbb{C}[\partial]$-module. Then $R$ is commutative if and only if $(\mathscr{L}ie_p\mbox{ }R)_\_$ is commutative.
\end{corollary}
\begin{proof}
	(i) From Proposition \ref{p3}, $a_{\lambda_1,\cdots,\lambda_p}=0$ if $a\in Tor\mbox{ }R$. If $a\notin Tor\mbox{ }R$, $a=\sum\limits_{i,j}c_{ij}\partial^iv^j+r$ where non-zero element $\sum\limits_{i,j}c_{ij}\partial^iv^j\in\mathbb{C}[\partial]\otimes V$ and $r\in Tor\mbox{ }R$. Since \begin{eqnarray}
	(\partial a)_{m_1,\cdots,m_p}=-(\sum_{i=0}^{p}m_ia_{m_1,\cdots,m_i-1,\cdots m_p}) 
	\end{eqnarray} holds in $\mathscr{L}ie_p\mbox{ }R$,  $a_{\lambda_1,\cdots,\lambda_p}\ne0$ for $m_i$ sufficiently large.
	
	(ii) It can be immediately deduced from (\ref{7}) and (\ref{9}).
\end{proof}
 Based on the above discussion, we can give the following theorem.
\begin{theorem}
	Suppose $R$ and $S$ are torsionless finite $n$-Lie conformal algebras. Then for every topological $n$-Lie algebra homomorphism $\varphi:(\mathscr{L}ie_p\mbox{ }R)_\_\to(\mathscr{L}ie_p\mbox{ }S)_\_ $ compatible with $\partial_{t_i}$-actions, there exists a unique homomorphism of $n$-Lie conformal algebras $\phi:R\to S$ inducing $\varphi$.
\end{theorem}
\begin{proof}
	Choose $r\in R$. As $S$ is torsionless, by (\ref{10}) and Corollary \ref{c2}, there is a unique $s\in S$ such that $\varphi(r_{m_1,\cdots,m_p})=s_{m_1,\cdots,m_p}$ for all $m_1,\cdots,m_p\in \mathbb{N}$. Let $\{s^i\}$ be a $\mathbb{C}[\partial]$-basis of $S$. Then there exist $c^i_{j_1,\cdots,j_p}\in\mathbb{C}$ such that
	$\varphi(r_{m_1,\cdots,m_p})=\sum\limits_{i,j_1,\cdots,j_p}c^i_{j_1,\cdots,j_p}s^i_{j_1,\cdots,j_p}$. Since $\varphi$ commutes with $\partial_{t_i}$, $j_l\le m_l$ for $l=1,\cdots,p$. Thus the relation can be rephrased as \begin{align}
	\varphi(r_{m_1,\cdots,m_p})=&\sum_{i}\sum_{j_1=0}^{m_1}\cdots\sum_{j_p=0}^{m_p}\tfrac{m_1!}{j_1!}\cdots\tfrac{m_p!}{j_p!}d^i_{m_1-j_1,\cdots,m_p-j_p}s^i_{j_1,\cdots,j_p}\notag\\
	=&\sum_{i}\sum_{j_1=0}^{m_1}\cdots\sum_{j_p=0}^{m_p}d^i_{m_1-j_1,\cdots,m_p-j_p}{\partial_{t_1}}^{m_1-j_1}\cdots{\partial_{t_p}}^{m_p-j_p}s^i_{m_1,\cdots,m_p}
	\end{align} with constants $d^i_{j_1,\cdots,j_p}$.
	Then  $\phi(r)=\sum\limits_{i,j_1,\cdots,j_p}d^i_{j_1,\cdots,j_p}{\partial_{t_1}}^{j_1}\cdots{\partial_{t_p}}^{j_p}s^i$ is the only map satisfies $(\phi(r))_{{m_1,\cdots,m_p}}=\varphi(r_{m_1,\cdots,m_p})$ for all $m_1,\cdots,m_p\in \mathbb{N}$, if $\phi(r)\in S$.
	
	Denote the origin of $(\mathscr{L}ie_p\mbox{ }R)_\_$ resp. $ (\mathscr{L}ie_p\mbox{ }S)_\_ $ by $\textbf{0}_R$ resp. $\textbf{0}_S$ and the nuclear base (\ref{11}) of $(\mathscr{L}ie_p\mbox{ }R)_\_$ resp. $ (\mathscr{L}ie_p\mbox{ }S)_\_ $ by $\mathscr{L}_p^{(m)}(R)$ resp. $\mathscr{L}_p^{(m)}(S)$. Since $\varphi(\textbf{0}_R)=\textbf{0}_S$ and $\varphi$ is continuous, for $\mathscr{L}_p^{(1)}(S)\ni\textbf{0}_S$, there exists a neighborhood $U$ of $\textbf{0}_R$ such that $\varphi(U)\subset\mathscr{L}_p^{(1)}(S)$. That is, $\varphi(\mathscr{L}_p^{(m)}(R))\subset\mathscr{L}_p^{(1)}(S)$ for $m$ sufficiently large, which means $d^i_{m_1,\cdots,m_p}=0$ for $m_1+\cdots+m_p\gg0$. So $\phi$ is a well-defined map from $R$ to $S$. The commutativity of $\varphi$ with $\partial_{t_i}$ brings about the compatibility of $\phi$ and $\partial$. By the definition of $\mathscr{L}ie_p\mbox{ }R$, $(\partial r)_{m_1,\cdots,m_p}=\partial (r_{m_1,\cdots,m_p})$. Then we have
	\begin{align*}
	(\phi(\partial r))_{m_1,\cdots,m_p}=\varphi((\partial r)_{m_1,\cdots,m_p})=&\varphi(\partial (r_{m_1,\cdots,m_p}))\\
	=&\partial\varphi(r_{m_1,\cdots,m_p})\\
	=&\partial((\phi(r))_{{m_1,\cdots,m_p}})\\
	=&(\partial (\phi(r)))_{m_1,\cdots,m_p}.
	\end{align*} Thus $\phi(\partial r)=\partial (\phi(r))$. Similarly, by (\ref{9}), $\phi$ commutes with the conformal brackets as $\varphi$ commutes with the $n$-Lie brackets. Hence $\phi$ is an $n$-Lie conformal algebra homomorphism between $R$ and $S$.
	
\end{proof}
	\section{Cohomology of $n$-Lie conformal algebra}\label{s4}
In this section, we establish the representation and cohomology theory of an $n$-Lie conformal algebra $R$ and explain their relation with that of $n$-Lie conformal algebras $ (\mathscr{L}ie_p\mbox{ }R)_\_ $.
	
	\begin{definition}
		\begin{em}
		Let $V$ and $W$ be $\mathbb{C}[\partial]$-modules. A $\mathbb{C}$-linear map $f^{\lambda_1,\cdots,\lambda_{n-2}}: V\to W$ is called a \textit{conformal linear map} if $${f^{\lambda_1,\cdots,\lambda_{n-2}}}_\lambda\partial v=(\partial+\lambda+\lambda_1+\cdots+\lambda_{n-2})({f^{\lambda_1,\cdots,\lambda_{n-2}}}_\lambda v)$$ for $v\in V$.
		\end{em}
	\end{definition}
Similar to \cite{ak}, we use $Chom(V,W)$ to denote the set of all conformal linear maps from $V$ to $W$ and $Cend\ V:=Chom(V,V)$. Then the generalized $Cend\ V$ also has a Lie conformal algebra structure defined by\begin{eqnarray}\label{12}
&{(\partial f^{\lambda_{1\to n-2}})}_\lambda v:=-\lambda{f^{\lambda_{1\to n-2}}}_\lambda v,\notag\\
&[{f^{\lambda_{1\to n-2}}}_\lambda{g^{\mu_{1\to n-2}}}]_\mu v:={f^{\lambda_{1\to n-2}}}_\lambda({g^{\mu_{1\to n-2}}}_{\mu-\lambda}v)-{g^{\mu_{1\to n-2}}}_{\mu-\lambda}({f^{\lambda_{1\to n-2}}}_\lambda v).
\end{eqnarray} Besides, it is obvious for an $n$-Lie conformal algebra $R$ that $L(\overrightarrow{a}):={a^1}_{\lambda_1}{\cdots}_{\lambda_{n-2}}a^{n-1}$ is an element in $Cend\ R$ by $$L(\overrightarrow{a})_\lambda r=[{a^1}_{\lambda_1}{\cdots}_{\lambda_{n-2}}{a^{n-1}}_\lambda r], r\in R.$$
\begin{definition}
	\begin{em}
		For an $n$-Lie conformal algebra $R$, $f^{\lambda_{1\to n-2}}\in Cend\ R$ is called a \textit{derivation} of $R$ if 
		\begin{align}\label{13}
		{f^{\lambda_{1\to n-2}}}_{\lambda}	[{a^1} _{\mu_1}{a^2}_{\mu_2}{\cdots}_{\mu_{n-1}}&a^n]=\sum_{i=1}^{n-1}[{a^1} _{\mu_1}\cdots{ a^{i-1}}_{\mu_{i-1}}({f^{\lambda_{1\to n-2}}}_{\lambda}{ a^i})_{\lambda+\lambda_{1\to n-2}+\mu_i}{ a^{i+1}}_{\mu_{i+1}}{\cdots}_{\mu_{n-1}}{a^n}]\notag\\
		&\qquad\qquad+[{a^1} _{\mu_1}{\cdots}_{\mu_{n-2}}{a^{n-1}}_{\mu_{n-1}}({f^{\lambda_{1\to n-2}}}_{\lambda}{ a^n})].
		\end{align} 
	\end{em}
\end{definition}

\begin{proposition}\label{p4}
	(i) Denote the set of all derivations of an $n$-Lie conformal algebra $R$ by $Der\ R$. Then $Der\ R$ is a subalgebra of $Cend\ R$.
	
	(ii) Every $L(\overrightarrow{a})$ is a derivation of $R$. And elements in the linear space spanned by such derivations are called \textbf{inner} derivations.
\end{proposition}

\begin{proof}
	(i) We only need to show $[{f^{\lambda_{1\to n-2}}}_\lambda{g^{\mu_{1\to n-2}}}]$ defined in (\ref{12}) is a derivation for all $f^{\lambda_{1\to n-2}},g^{\mu_{1\to n-2}}\in Der\ R$. By definition, 
	\begin{align*}
	&[{f^{\lambda_{1\to n-2}}}_\lambda{g^{\mu_{1\to n-2}}}]_\mu[{a^1} _{\xi_1}{a^2}_{\xi_2}{\cdots}_{\xi_{n-1}}a^n]\\
	=&{f^{\lambda_{1\to n-2}}}_\lambda({g^{\mu_{1\to n-2}}}_{\mu-\lambda}[{a^1} _{\xi_1}{a^2}_{\xi_2}{\cdots}_{\xi_{n-1}}a^n])-{g^{\mu_{1\to n-2}}}_{\mu-\lambda}({f^{\lambda_{1\to n-2}}}_\lambda [{a^1} _{\xi_1}{a^2}_{\xi_2}{\cdots}_{\xi_{n-1}}a^n]).\\
	\end{align*}
	Substituting in (\ref{13}), we can obtain a lengthy expansion. Then it can be noted that the terms in the form of $$[{a^1} _{\xi_1}\cdots({f^{\lambda_{1\to n-2}}}_\lambda a^j)_{\lambda+\lambda_{1\to n-2}+\xi_j}\cdots({g^{\mu_{1\to n-2}}}_{\mu-\lambda}{ a^i})_{\mu-\lambda+\mu_{1\to n-2}+\xi_i}{\cdots}_{\xi_{n-1}}{a^n}]$$ occur in pairs with opposite signs and thus cancel each other out. With the remaining items, we have \begin{align*}
	&{f^{\lambda_{1\to n-2}}}_\lambda({g^{\mu_{1\to n-2}}}_{\mu-\lambda}[{a^1} _{\xi_1}{a^2}_{\xi_2}{\cdots}_{\xi_{n-1}}a^n])-{g^{\mu_{1\to n-2}}}_{\mu-\lambda}({f^{\lambda_{1\to n-2}}}_\lambda [{a^1} _{\xi_1}{a^2}_{\xi_2}{\cdots}_{\xi_{n-1}}a^n])\\
	=&\sum_{i=1}^{n-1}[{a^1} _{\xi_1}\cdots{ a^{i-1}}_{\xi_{i-1}}({f^{\lambda_{1\to n-2}}}_\lambda({g^{\mu_{1\to n-2}}}_{\mu-\lambda}{ a^i}))_{\mu+\lambda_{1\to n-2}+\mu_{1\to n-2}+\xi_i}{ a^{i+1}}_{\xi_{i+1}}{\cdots}_{\xi_{n-1}}{a^n}]\notag\\
	&+[{a^1} _{\xi_1}{\cdots}_{\xi_{n-2}}{a^{n-1}}_{\xi_{n-1}}({f^{\lambda_{1\to n-2}}}_\lambda({g^{\mu_{1\to n-2}}}_{\mu-\lambda}{ a^n}))]\\
	&-\sum_{i=1}^{n-1}[{a^1} _{\xi_1}\cdots{ a^{i-1}}_{\xi_{i-1}}({g^{\mu_{1\to n-2}}}_{\mu-\lambda}({f^{\lambda_{1\to n-2}}}_{\lambda}{ a^i}))_{\mu+\lambda_{1\to n-2}+\mu_{1\to n-2}+\xi_i}{ a^{i+1}}_{\xi_{i+1}}{\cdots}_{\xi_{n-1}}{a^n}]\notag\\
	&-[{a^1} _{\xi_1}{\cdots}_{\xi_{n-2}}{a^{n-1}}_{\xi_{n-1}}({g^{\mu_{1\to n-2}}}_{\mu-\lambda}({f^{\lambda_{1\to n-2}}}_{\lambda}{ a^n}))].
	\end{align*} This completes the proof.
	
	(ii) To prove the statement, we need the following key formulas. For any $f^{\lambda_{1\to n-2}},g^{\mu_{1\to n-2}}\in Cend\ R, a\in R$, \begin{align}
		{f^{\lambda_{1\to n-2}}}_{-\partial-\lambda}({g^{\mu_{1\to n-2}}}_{-\partial-\mu}a)=&\sum_{m}{f^{\lambda_{1\to n-2}}}_{-\partial-\lambda}((-\partial-\mu)^{(m)}({g^{\mu_{1\to n-2}}}_{(m)}a))\notag\\
		=&\sum_{m}(\lambda-\mu-\lambda_{1\to n-2})^{(m)}({f^{\lambda_{1\to n-2}}}_{-\partial-\lambda}({g^{\mu_{1\to n-2}}}_{(m)}a))\notag\\
		=&{f^{\lambda_{1\to n-2}}}_{-\partial-\lambda}({g^{\mu_{1\to n-2}}}_{\lambda-\mu-\lambda_{1\to n-2}}a),
	\end{align} 
	and \begin{align}
	&{f^{\lambda_{1\to n-2}}}_{\lambda}({g^{\mu_{1\to n-2}}}_{\mu}a)={f^{\lambda_{1\to n-2}}}_{-\partial-\xi}({g^{\mu_{1\to n-2}}}_{-\partial-\eta}a)\notag\\
	=&{f^{\lambda_{1\to n-2}}}_{-\partial-\xi}({g^{\mu_{1\to n-2}}}_{\xi-\eta-\lambda_{1\to n-2}}a)={f^{\lambda_{1\to n-2}}}_{\lambda}({g^{\mu_{1\to n-2}}}_{\mu-\lambda-\lambda_{1\to n-2}}a).
	\end{align}
	
	Let $L(\overrightarrow{b})={b^1}_{\lambda_1}{\cdots}_{\lambda_{n-2}}b^{n-1}$ be an inner derivation of $R$. Then from (C3)-(C4) and the above formulas, we have 
	\begin{align*}
	&{L(\overrightarrow{b})}_{\lambda}	[{a^1} _{\mu_1}{a^2}_{\mu_2}{\cdots}_{\mu_{n-1}}a^n]\notag\\
	=&\sum_{i=1}^{n-1}(-1)^{n-i}[{a^1} _{\mu_1}\cdots{ a^{i-1}}_{\mu_{i-1}}{ a^{i+1}}_{\mu_{i+1}}{\cdots}_{\mu_{n-1}}{a^n}_{-\partial-\lambda-\lambda_{1\to n-2}-\mu_{1\to n-1}}({L(\overrightarrow{b})}_{-\partial-\mu_i-\lambda_{1\to n-2}}{ a^i})]\notag\\
	&+[{a^1} _{\mu_1}{\cdots}_{\mu_{n-2}}{a^{n-1}}_{\mu_{n-1}}({L(\overrightarrow{b})}_{\lambda+\mu_{1\to n-1}}{ a^n})]\\
	=&\sum_{i=1}^{n-1}(-1)^{n-i}[{a^1} _{\mu_1}\cdots{ a^{i-1}}_{\mu_{i-1}}{ a^{i+1}}_{\mu_{i+1}}{\cdots}_{\mu_{n-1}}{a^n}_{-\partial-\lambda-\lambda_{1\to n-2}-\mu_{1\to n-1}}({L(\overrightarrow{b})}_{\lambda}{ a^i})]\notag\\
	&+[{a^1} _{\mu_1}{\cdots}_{\mu_{n-2}}{a^{n-1}}_{\mu_{n-1}}({L(\overrightarrow{b})}_{\lambda}{ a^n})]\\
	=&\sum_{i=1}^{n-1}[{a^1} _{\mu_1}\cdots{ a^{i-1}}_{\mu_{i-1}}({L(\overrightarrow{b})}_{\lambda}{ a^i})_{\lambda+\lambda_{1\to n-2}+\mu_i}{ a^{i+1}}_{\mu_{i+1}}{\cdots}_{\mu_{n-1}}{a^n}]\notag\\
	&+[{a^1} _{\mu_1}{\cdots}_{\mu_{n-2}}{a^{n-1}}_{\mu_{n-1}}({L(\overrightarrow{b})}_{\lambda}{ a^n})].
	\end{align*}
\end{proof}
\begin{remark}
	It can be easily computed that \begin{align}\label{17}
	[L(\overrightarrow{a})_\lambda L(\overrightarrow{b})]_\mu r=&\sum_{i=1}^{n-2}({b^1} _{\mu_1}\cdots{ b^{i-1}}_{\mu_{i-1}}({L(\overrightarrow{a})}_{\lambda}{ b^i})_{\lambda+\lambda_{1\to n-2}+\mu_i}{ b^{i+1}}_{\mu_{i+1}}\cdots{b^{n-1}})_{\mu-\lambda}r\notag\\
	&+({b^1} _{\mu_1}\cdots{ b^{n-2}}_{\mu_{n-2}}({L(\overrightarrow{a})}_{\lambda}{ b^{n-1}}))_{\lambda_{1\to n-2}+\mu}r
	\end{align} for $L(\overrightarrow{a})={a^1}_{\lambda_1}{\cdots}_{\lambda_{n-2}}a^{n-1},L(\overrightarrow{b})={b^1}_{\mu_1}{\cdots}_{\mu_{n-2}}b^{n-1},r\in R$.
	
\end{remark}
	\begin{definition}
		\begin{em}
			Let $R$ be an $n$-Lie conformal algebra. A $\mathbb{C}[\partial]$-module $M$ is called a \textit{conformal representation} of $R$ if it is endowed with the $\{\lambda_{1\to n-1}\}$-action  ${a^1} _{\lambda_1}{\cdots}_{\lambda_{n-2}}{a^{n-1}}_{\lambda_{n-1}}v$ satisfying
			
			(M1) ${a^1} _{\lambda_1}{\cdots}_{\lambda_{n-2}}{a^{n-1}}\in Cend\ M$, i.e. ${a^1} _{\lambda_1}{\cdots}_{\lambda_{n-2}}{a^{n-1}}_{\lambda_{n-1}}v\in \mathbb{C}[\partial,\lambda_1,\lambda_2,\cdots,\lambda_{n-1}]\otimes_{\mathbb{C}[\partial]}M$,
			
			(M2) ${a^1} _{\lambda_1}{\cdots}_{\lambda_{i-1}}{\partial a^i}_{\lambda_i}{\cdots}_{\lambda_{n-2}}{a^{n-1}}_{\lambda_{n-1}}v=-\lambda_i{a^1} _{\lambda_1}{\cdots}_{\lambda_{n-2}}{a^{n-1}}_{\lambda_{n-1}}v$ for $i=1,2,\cdots,n-1$,
			
			\qquad\mbox{ }$[{a^1} _{\lambda_1}{\cdots}_{\lambda_{n-2}}{a^{n-1}}_{\lambda_{n-1}}\partial v]=(\partial+\lambda_1+\lambda_2+\cdots+\lambda_{n-1}){a^1} _{\lambda_1}{\cdots}_{\lambda_{n-2}}{a^{n-1}}_{\lambda_{n-1}}v$,
			
			(M3) ${a^1} _{\lambda_1}\cdots{ a^i}_{\lambda_i}{ a^{i+1}}_{\lambda_{i+1}}{\cdots}_{\lambda_{n-1}}v=-{a^1} _{\lambda_1}\cdots{ a^{i+1}}_{\lambda_{i+1}}{ a^i}_{\lambda_i}{\cdots}_{\lambda_{n-1}}v$ for $i=1,2,\cdots,n-2$,

			(M4) $\sum\limits_{i=1}^{n}(-1)^{n-i}{a^1} _{\lambda_1}\cdots{ a^{i-1}}_{\lambda_{i-1}}{ a^{i+1}}_{\lambda_{i+1}}{\cdots}_{\lambda_{n-1}}{a^n}_{\lambda_n}({ a^i}_{\lambda_i}{b^1}_{\lambda_{n+1}}{\cdots}_{\lambda_{2n-3}}{b^{n-2}}_{\lambda_{2n-2}}v)$
			
			\qquad$=[{a^1} _{\lambda_1}{a^2}_{\lambda_2}{\cdots}_{\lambda_{n-1}}a^n]_{\lambda_1+\cdots+\lambda_n}{b^1}_{\lambda_{n+1}}{\cdots}_{\lambda_{2n-3}}{b^{n-2}}_{\lambda_{2n-2}}v$,
			
			\qquad${a^1}_{\lambda_1}{\cdots}_{\lambda_{n-2}}{a^{n-1}}_{\lambda}({b^1}_{\mu_1}{\cdots}_{\mu_{n-2}}{b^{n-1}}_{\mu }v)-{b^1}_{\mu_1}{\cdots}_{\mu_{n-2}}{b^{n-1}}_{\mu}({a^1}_{\lambda_1}{\cdots}_{\lambda_{n-2}}{a^{n-1}}_{\lambda}v)$
			
			\qquad$=[({a^1}_{\lambda_1}{\cdots}_{\lambda_{n-2}}a^{n-1})_\lambda ({b^1}_{\mu_1}{\cdots}_{\mu_{n-2}}b^{n-1})]_{\lambda+\mu} v,$\\
			for all $a^1,\cdots,a^n,b^1,\cdots,b^{n-1} \in R,v\in M$.
		\end{em}
	\end{definition}
\begin{remark}
	In essence, the notion of representations is given by replacing the last element in $\{\lambda_{1\to n-1}\}$-bracket of $R$ by an element of $M$. Since it is a left module action, we lose part of skew-symmetry in (C3). Due to the two different positions of elements in the Filippov identity (C4), there are two equations in (M4).  
\end{remark}
From (C4) and the above definition, the following lemma can be deduced.
\begin{lemma}
Suppose $M$ is a conformal representation of $n$-Lie conformal algebra $R$. Then $R\oplus M$ is an $n$-Lie conformal algebra with the conformal brackets defined as
\begin{align}\label{14}
&[{a^1+v^1} _{\lambda_1}{a^2+v^2}_{\lambda_2}{\cdots}_{\lambda_{n-1}}{a^n+v^n}]\notag\\
=&[{a^1} _{\lambda_1}{\cdots}_{\lambda_{n-1}}a^n]+\sum_{i=1}^{n-1}(-1)^{n-i}{a^1} _{\lambda_1}\cdots{ a^i}_{\lambda_i}{ a^{i+1}}_{\lambda_{i+1}}{\cdots}_{\lambda_{n-1}}{a^n}_{-\partial-\sum\lambda_i}v^i+[{a^1} _{\lambda_1}\cdots{a^{n-1}}_{\lambda_{n-1}}v^n].
\end{align} 
\end{lemma}
\begin{proof}
	Our goal is to show (C1)-(C4) hold for (\ref{14}). In fact, (C1)-(C3) of $R\oplus M$ can be immediately deduced from (C1)-(C3) of $R$ and(M1)-(M3). A slightly more complex discussion is given for the verification of \begin{align}\label{15}
	&[[{a^1+v^1} _{\lambda_1}{\cdots}_{\lambda_{n-1}}{a^n+v^n}]_{\lambda_1+\cdots+\lambda_n}{b^2+w^2}_{\lambda_{n+1}}{\cdots}_{\lambda_{2n-2}}{b^n+w^n}]\notag\\
	=&\sum\limits_{i=1}^{n}(-1)^{n-i}[{a^1+v^1} _{\lambda_1}\cdots{ a^{i-1}+v^{i-1}}_{\lambda_{i-1}}{ a^{i+1}+v^{i+1}}_{\lambda_{i+1}}{\cdots}_{\lambda_{n-1}}{a^n+v^n}_{\lambda_n}\notag\\
	&\qquad[{ a^i+v^i}_{\lambda_i}{b^2+w^2}_{\lambda_{n+1}}{\cdots}_{\lambda_{2n-2}}{b^n+w^n}]],
	\end{align}where $a^1,\cdots,a^n,b^2,\cdots,b^n \in R$, $v^1,\cdots,v^n,w^2,\cdots,w^n \in M$.
	Putting in (\ref{14}), we obtain the expansion of two sides of (\ref{15}). The summands are in the following three forms: (a) terms in $R$, (b) terms contained $v^i$, (c) terms contained $w^i$. Terms of type (a) in two sides are equal due to (C4) of $R$. The equivalences of terms of type (a) in two sides can be inferred from (M3) and the last equation in (M4), while that of terms of type (b) can be inferred from (M3) and the first equation in (M4).   
\end{proof}
\begin{lemma}
	Let $R$ be an $n$-Lie conformal algebra. If $M$ is a conformal representation of conformal algebra $R$, then $\mathscr{V}_p(M):=M[t_1^{\pm1},\cdots,t_p^{\pm1}]/\tilde{\partial}_pM[t_1^{\pm1},\cdots,t_p^{\pm1}]$ is a representation of $n$-Lie algebra $\mathscr{L}ie_p\mbox{ }R$. Furthermore, $M$ carries a $(\mathscr{L}ie_p\mbox{ }R)_\_$-representation structure.
\end{lemma}
\begin{proof}
	A discussion similar to Lemma \ref{p1} will show $M[t_1^{\pm1},\cdots,t_p^{\pm1}]$ is a conformal representation of $\tilde{R}_p$. Then by (M2), a similar proof of Corollary \ref{c1} can explain $\mathscr{V}_p(M)$ is a representation of $n$-Lie algebra $\mathscr{L}ie_p\mbox{ }R$. Explicitly, the $\mathscr{L}ie_p\mbox{ }R$--representation structure on $\mathscr{V}_p(M)$ is given by \begin{align}
	(a^1_{m_1^1,\cdots,m_p^1},&\cdots,a^{n-1}_{m_1^{n-1},\cdots,m_p^{n-1}})v_{m_1^n,\cdots,m_p^n}=\sum_{j_{1\to p}^{1\to {n-1}}}\tbinom{m_1^1}{j_1^1}\cdots\tbinom{m_p^{n-1}}{j_p^{n-1}}\notag\\
	&({a^1}_{(j^1_1+\cdots+j_p^1)}\cdots{a^{n-1}}_{(j_1^{n-1}+\cdots+j_p^{n-1})}v)_{m_1^1+\cdots+m_1^{n}-j_1^1-\cdots-j_1^{n-1},\cdots,m_p^1+\cdots+m_p^{n}-j_p^1-\cdots-j_p^{n-1}},
	\end{align}where $v_{m_1,\cdots,m_p}:=v\otimes t_1^{m_1}\cdots t_p^{m_p}\in\mathscr{V}_p(M)$.  And the induced $(\mathscr{L}ie_p\mbox{ }R)_\_$-representation structure on $M$ is defined as\begin{align}
	(a^1_{m_1^1,\cdots,m_p^1},&\cdots,a^{n-1}_{m_1^{n-1},\cdots,m_p^{n-1}})v=({a^1}_{(m^1_1+\cdots+m_p^1)}\cdots{a^{n-1}}_{(m_1^{n-1}+\cdots+m_p^{n-1})}v).
	\end{align}
\end{proof}
 Now it is time to construct the basic complex of $R$. Let $R$ be an $n$-Lie conformal algebra and $M$ be a conformal representation of $R$. Define $\widetilde{C}^q(R,M)$ as the set of $\mathbb{C}$-linear maps $$ \gamma: R^{\otimes^{n-1}}\otimes {\tiny \begin{matrix}
 	(q-1)\\\cdots
 	\end{matrix}} \otimes R^{\otimes^{n-1}}\otimes R\to M[\lambda^{1\to p-1}_{1\to n-1},\lambda^{p-1}_n] $$ 
 $$a^1_{1\to n-1}\otimes\cdots\otimes a^{q-1}_{1\to n-1}\otimes a^{q-1}_n\mapsto \gamma_{\lambda^{1\to q-1}_{1\to n-1},\lambda^{q-1}_n}(a^1_{1\to n-1},\cdots, a^{q-1}_{1\to n-1}, a^{q-1}_n)$$satisfying the following conditions.
 
 (B1) \textbf{Conformal antilinearity.}\begin{align*}
&\gamma_{\lambda^{1\to q-1}_{1\to n-1},\lambda^{q-1}_n}(a^1_{1\to n-1},\cdots,a^i_1,\cdots,\partial a^i_j,\cdots,a^i_{n-1},\cdots, a^{q-1}_{1\to n-1}, a^{q-1}_n)\\
&\qquad\qquad\qquad\qquad\qquad\qquad=-\lambda^i_j\gamma_{\lambda^{1\to q-1}_{1\to n},\lambda^{q-1}_n}(a^1_{1\to n-1},\cdots, a^{q-1}_{1\to n-1}, a^{q-1}_n).
 \end{align*}
 
 (B2) \textbf{Skew-symmetry.} For each fixed $i$, $\gamma$ is skew-symmetric with respect to simultaneous permutations of $a^i_j$ and $\lambda^i_j$.
\begin{remark}
	The image be interpreted as the result of $\gamma$-action on ${a^1_{1}}_{\lambda^1_{1}}\cdots {a^i_j}_{\lambda^i_{j}}\cdots {a^{q-1}_{n-1}}_{\lambda^{q-1}_{n-1}} {a^{q-1}_n}_{\lambda^{q-1}_n}$.
\end{remark}

\begin{proposition}\label{p5}
	For $\gamma \in \widetilde{C}^q(R,M)$, \begin{align}\label{16}
	&(D\gamma)_{\lambda^{1\to q}_{1\to n-1},\lambda^q_n}(a^1_{1\to n-1},\cdots, a^q_{1\to n-1}, a^q_n)\notag\\
	=&\sum_{1\le i<j\le q}(-1)^{i}\gamma_{\lambda^{1\to\hat{i}\to j-1}_{1\to n-1},[\lambda^i_{1\to n-1},\lambda^j_{1\to n-1}],\lambda^{j+1\to q}_{1\to n-1},\lambda^q_n}(a^1_{1\to n-1},\cdots,\hat{a}^i_{1\to n-1},\cdots,a^{j-1}_{1\to n-1},\notag
	\\&\qquad\qquad\qquad\qquad\qquad\qquad\qquad\qquad\qquad\qquad\qquad[{a^i_{1\to n-1}},a^j_{1\to n-1}],a^{j+1}_{1\to n-1},\cdots,a^q_n)\notag\\
	&+\sum_{i=1}^{q}(-1)^{i}\gamma_{\lambda^{1\to\hat{i}\to q}_{1\to n-1},\lambda^q_n+\lambda^i_{1\to n-1}}(a^1_{1\to n-1},\cdots,\hat{a}^i_{1\to n-1},\cdots,{a}^q_{1\to n-1},[{a^i_1}_{\lambda^i_1}{\cdots}_{\lambda^i_{n-2}}{a^i_{n-1}}_{\lambda^i_{n-1}}a^q_n])\notag\\
	&+\sum_{i=1}^{q}(-1)^{i+1}{a^i_1}_{\lambda^i_1}{\cdots}_{\lambda^i_{n-2}}{a^i_{n-1}}_{\lambda^i_{n-1}}(\gamma_{\lambda^{1\to\hat{i}\to q}_{1\to n-1},\lambda^q_n}(a^1_{1\to n-1},\cdots,\hat{a}^i_{1\to n-1},\cdots,{a}^q_{1\to n-1},a^q_n))\notag\\
	&+\sum_{i=1}^{n-1}(-1)^{n+q-i+1}{a^q_1}_{\lambda^q_1}\cdots{\hat{a}^q_i}\cdots{a^q_{n-1}}_{\lambda^q_{n-1}} {a^q_n}_{\lambda^q_n}(\gamma_{\lambda^{1\to q-1}_{1\to n-1},\lambda^q_i}(a^1_{1\to n-1},\cdots, a^{q-1}_{1\to n-1}, a^q_i)),
	\end{align}where $[{a^i_{1\to n-1}},a^j_{1\to n-1}]$ and $[\lambda^i_{1\to n-1},\lambda^j_{1\to n-1}]$ are the corresponding elements and indexes of the expansion of $[({a^i_1}_{\lambda^i_1}{\cdots}_{\lambda^i_{n-2}}a^i_{n-1})_{\lambda^i_{n-1}}({a^j_1}_{\lambda^j_1}{\cdots}_{\lambda^j_{n-2}}a^j_{n-1})]_{\lambda^i_{n-1}+\lambda^j_{n-1}}$, 
	defines a map $D$ from $\widetilde{C}^q(R,M)$ to $\widetilde{C}^{q+1}(R,M)$ satisfying $D^2=0$.
\end{proposition}
\begin{proof}
	\textbf{Step 1:} $D\gamma$ meets the conformal antilinearity condition. We first consider $\partial a^i_j$ with $i=1,\cdots,q,j=1,\cdots,n-1$ and then concentrate on $\partial a^p_n$.
	
	With the $\partial$-action on $a^i_j$, by (\ref{17}) and (B1) of $\gamma$, the value of the first summand of (\ref{16}) equals the original value multiplying $-\lambda^i_j$. Similar changes are seen in the second and third summands by (C2) of $R$ and (M2) of $M$. For the fourth summand, if $i\ne q$, we can get the multiple $-\lambda^i_j$ by (B1) of $\gamma$ and if $i=q$, the multiple $-\lambda^q_j$ is given by (B1) of $\gamma$ and (M2) of $M$.
	
	Then let we focus on $\partial a^q_n$. By (B1) of $\gamma$, we can extract the multiple $-\lambda^q_n$ from the first and third summands. The $-\lambda^q_n$ of the second summand is derived from the joint action of (C2) of $R$ and (B1) of $\gamma$, while that of the fourth summand is derived from (M2) of $M$.
	
	\textbf{Step 2:} $D\gamma$ also satisfies (B2) so that $D$ is a well-defined map between $\widetilde{C}^q(R,M)$ and $\widetilde{C}^{q+1}(R,M)$. Analogously, we discuss the transposition of $a^i_j$ and $a^i_{j+1}$ with $i=1,\cdots,q,j=1,\cdots,n-2$ and the transposition of $a^q_{n-1}$ and $a^q_n$ separately.
	
	Since\begin{align*}
	&[({a^i_1}_{\lambda^i_1}{\cdots}_{\lambda^i_{n-2}}a^i_{n-1})_{\lambda^i_{n-1}}({a^j_1}_{\lambda^j_1}{\cdots}_{\lambda^j_{n-2}}a^j_{n-1})]_{\lambda^i_{n-1}+\lambda^j_{n-1}}v\\
	=&{a^i_1}_{\lambda^i_1}{\cdots}_{\lambda^i_{n-2}}{a^i_{n-1}}_{\lambda^i_{n-1}}({a^j_1}_{\lambda^j_1}{\cdots}_{\lambda^j_{n-2}}{a^j_{n-1}}_{\lambda^j_{n-1}}v)-{a^j_1}_{\lambda^j_1}{\cdots}_{\lambda^j_{n-2}}{a^j_{n-1}}_{\lambda^j_{n-1}}({a^i_1}_{\lambda^i_1}{\cdots}_{\lambda^i_{n-2}}{a^i_{n-1}}_{\lambda^i_{n-1}}v),
	\end{align*} 
	the transposition of $a^i_j$ and $a^i_{j+1}$ changes the sign of every term in the first summand. By (C3), (M3) and (B2), it is easy to see that the signs of the second and third summands are changed. For the fourth summand, if $i\ne q$, the sign is changed by (B2) and if $i= q$, the sign is changed by $(-1)^{n+q-i+1}$.
	
	Then we consider \begin{align*}
	&(D\gamma)_{\lambda^{1\to q}_{1\to n-1},\lambda^q_n}(a^1_{1\to n-1},\cdots, a^{q-1}_{1\to n-1}, a^q_1,\cdots,a^q_{n-1},a^q_n)\\
	+&(D\gamma)_{\lambda^{1\to q-1}_{1\to n-1},\lambda^q_{1\to n-2},\lambda^q_n,\lambda^q_{n-1}}(a^1_{1\to n-1},\cdots, a^{q-1}_{1\to n-1}, a^q_1,\cdots,a^q_n,a^q_{n-1}).
	\end{align*}
	As discussed above, we only need to consider the case when $a^p_{n-1}$ and $a^p_n$ are separated. Putting in (\ref{17}), we find terms in the first and second summand occur in pairs with opposite signs and thus cancel each other out. And this also happens in the third and fourth summand.  
	
	\textbf{Step 3:} Now we are going to show $D^2=0$. Substituting in (\ref{16}), we have
	\begin{align}\label{18}
	&(D^2\gamma)_{\lambda^{1\to q}_{1\to n-1},\lambda^q_n}(a^1_{1\to n-1},\cdots, a^q_{1\to n-1}, a^q_n)\notag\\
	=&\sum_{1\le i<j\le q}(-1)^{i}(D\gamma)_{\lambda^{1\to\hat{i}\to j-1}_{1\to n-1},[\lambda^i_{1\to n-1},\lambda^j_{1\to n-1}],\lambda^{j+1\to q}_{1\to n-1},\lambda^q_n}(a^1_{1\to n-1},\cdots,\hat{a}^i_{1\to n-1},\cdots,a^{j-1}_{1\to n-1},\notag
	\\&\qquad\qquad\qquad\qquad\qquad\qquad\qquad\qquad\qquad\qquad\qquad[{a^i_{1\to n-1}},a^j_{1\to n-1}],a^{j+1}_{1\to n-1},\cdots,a^q_n)\notag\\
	&+\sum_{i=1}^{q}(-1)^{i}(D\gamma)_{\lambda^{1\to\hat{i}\to q}_{1\to n-1},\lambda^q_n+\lambda^i_{1\to n-1}}(a^1_{1\to n-1},\cdots,\hat{a}^i_{1\to n-1},\cdots,{a}^q_{1\to n-1},[{a^i_1}_{\lambda^i_1}{\cdots}_{\lambda^i_{n-2}}{a^i_{n-1}}_{\lambda^i_{n-1}}a^q_n])\notag\\
	&+\sum_{i=1}^{q}(-1)^{i+1}{a^i_1}_{\lambda^i_1}{\cdots}_{\lambda^i_{n-2}}{a^i_{n-1}}_{\lambda^i_{n-1}}((D\gamma)_{\lambda^{1\to\hat{i}\to q}_{1\to n-1},\lambda^q_n}(a^1_{1\to n-1},\cdots,\hat{a}^i_{1\to n-1},\cdots,{a}^q_{1\to n-1},a^q_n))\notag\\
	&+\sum_{i=1}^{n-1}(-1)^{n+q-i+1}{a^q_1}_{\lambda^q_1}\cdots{\hat{a}^q_i}\cdots{a^q_{n-1}}_{\lambda^q_{n-1}} {a^q_n}_{\lambda^q_n}((D\gamma)_{\lambda^{1\to q-1}_{1\to n-1},\lambda^q_i}(a^1_{1\to n-1},\cdots, a^{q-1}_{1\to n-1}, a^q_i))
	\end{align}
	Then we discuss each summand in the expansion of (\ref{18}) in turn. For the sake of simplicity, we use $S_{ij}$, $i,j=1,2,3,4$ to mean the $j$th summand of the expansion of the $i$th summand in (\ref{18}).
	
	\textbf{($S_{11}$).} Terms in $S_{11}$ can be divided into the following three types. 
	\begin{align*}
	&(a)\ \gamma_{\cdots,[\lambda^i_{1\to n-1},\lambda^j_{1\to n-1}],\cdots,[\lambda^k_{1\to n-1},\lambda^l_{1\to n-1}],\cdots,\lambda^q_n}(\cdots,[{a^i_{1\to n-1}},a^j_{1\to n-1}],\\
	&\qquad\qquad\qquad\qquad\qquad\qquad\qquad\qquad\qquad\qquad\cdots,[{a^k_{1\to n-1}},a^l_{1\to n-1}],\cdots,a^q_n), \\
	&(b)\ \gamma_{\cdots,[\lambda^i_{1\to n-1},[\lambda^j_{1\to n-1},\lambda^k_{1\to n-1}]],\cdots,\lambda^q_n}(\cdots,[{a^i_{1\to n-1}},[{a^j_{1\to n-1}},a^k_{1\to n-1}]],\cdots,a^q_n),\\
	&(c)\ \gamma_{\cdots,[[\lambda^i_{1\to n-1},\lambda^j_{1\to n-1}],\lambda^k_{1\to n-1}],\cdots,\lambda^q_n}(\cdots,[[{a^i_{1\to n-1}},{a^j_{1\to n-1}}],a^k_{1\to n-1}],\cdots,a^q_n).
	\end{align*}
	
	For fixed $i,j,k,l$, terms of type (a) will appear twice because the advents of $[{a^i_{1\to n-1}},a^j_{1\to n-1}]$ and $[{a^k_{1\to n-1}},a^l_{1\to n-1}]$ have different order. In the case when $i<k$, if $[{a^i_{1\to n-1}},a^j_{1\to n-1}]$ appears first, the sign would be $(-1)^{i+k-1}$, while if $[{a^k_{1\to n-1}},a^l_{1\to n-1}]$ appears first, the sign would be $(-1)^{i+k}$. Then they cancel each other out. Similar things happen in the case when $k<i$.
	
	Each term of type (c) appears if and only if $i<j<k$, while each term of type (b) appears if and only if $i<k$ and $j<k$. So each term of type (c) can be matched with two terms of type (b) in the following way.\begin{align*}
	&(-1)^{i+j-1}\gamma_{\cdots,[[\lambda^i_{1\to n-1},\lambda^j_{1\to n-1}],\lambda^k_{1\to n-1}],\cdots,\lambda^q_n}(\cdots,[[{a^i_{1\to n-1}},{a^j_{1\to n-1}}],a^k_{1\to n-1}],\cdots,a^q_n)\\
	+&(-1)^{i+j}\gamma_{\cdots,[\lambda^i_{1\to n-1},[\lambda^j_{1\to n-1},\lambda^k_{1\to n-1}]],\cdots,\lambda^q_n}(\cdots,[{a^i_{1\to n-1}},[{a^j_{1\to n-1}},a^k_{1\to n-1}]],\cdots,a^q_n)\\
	+&(-1)^{i+j-1}\gamma_{\cdots,[\lambda^j_{1\to n-1},[\lambda^i_{1\to n-1},\lambda^k_{1\to n-1}]],\cdots,\lambda^q_n}(\cdots,[{a^j_{1\to n-1}},[{a^i_{1\to n-1}},a^k_{1\to n-1}]],\cdots,a^q_n).
	\end{align*} Denote $L(\overrightarrow{a^l}):={a^l_1}_{\lambda^l_1}{\cdots}_{\lambda^l_{n-2}}a^l_{n-1}$. Then the different parts of the three summands correspond to $$[[{L(\overrightarrow{a^i})}_{\lambda^i_{n-1}}L(\overrightarrow{a^j})]_{\lambda^i_{n-1}+\lambda^j_{n-1}}L(\overrightarrow{a^k})]_{\lambda^i_{n-1}+\lambda^j_{n-1}+\lambda^k_{n-1}},$$
	$$-[{L(\overrightarrow{a^i})}_{\lambda^i_{n-1}}[{L(\overrightarrow{a^j})}_{\lambda^j_{n-1}}L(\overrightarrow{a^k})]]_{\lambda^i_{n-1}+\lambda^j_{n-1}+\lambda^k_{n-1}},$$
	and 
	$$[{L(\overrightarrow{a^j})}_{\lambda^j_{n-1}}[{L(\overrightarrow{a^i})}_{\lambda^i_{n-1}}L(\overrightarrow{a^k})]]_{\lambda^i_{n-1}+\lambda^j_{n-1}+\lambda^k_{n-1}}.$$
	By Proposition \ref*{p4} and linearity of $\gamma$, the sum is zero for each $i<j<k$.
	
	\textbf{($S_{12}$, $S_{21}$ and $S_{22}$).} Terms in $S_{12}$ have two types.\begin{align*}
	&(a)\ \gamma_{\cdots,[\lambda^i_{1\to n-1},\lambda^j_{1\to n-1}],\cdots,\lambda^q_n+\lambda^k_{1\to n-1}}(\cdots,[{a^i_{1\to n-1}},a^j_{1\to n-1}],\cdots,[{a^k_1}_{\lambda^k_1}{\cdots}_{\lambda^k_{n-2}}{a^k_{n-1}}_{\lambda^k_{n-1}}a^q_n]),\\
	&(b)\ \gamma_{\cdots,\lambda^q_n+\lambda^i_{1\to n-1}+\lambda^j_{1\to n-1}}(\cdots,[{L(\overrightarrow{a^i})}_{\lambda^i_{n-1}}L(\overrightarrow{a^j})]_{\lambda^i_{n-1}+\lambda^j_{n-1}}a^q_n).
	\end{align*}
	
	Terms of type (a) in $S_{12}$ and terms of $S_{21}$ are in the same form. But they have opposite signs in different summands whether $i<k$ or $i>k$ for fixed $i,j,k$. So they are offset. On the other hand, each term of type (b) can be matched with two terms of $S_{22}$ as follows:\begin{align*}
	&(-1)^{i+j-1}\gamma_{\cdots,\lambda^q_n+\lambda^i_{1\to n-1}+\lambda^j_{1\to n-1}}(\cdots,[{L(\overrightarrow{a^i})}_{\lambda^i_{n-1}}L(\overrightarrow{a^j})]_{\lambda^i_{n-1}+\lambda^j_{n-1}}a^q_n)\\
	+&(-1)^{i+j}\gamma_{\cdots,\lambda^q_n+\lambda^i_{1\to n-1}+\lambda^j_{1\to n-1}}(\cdots,{L(\overrightarrow{a^i})}_{\lambda^i_{n-1}}({L(\overrightarrow{a^j})}_{\lambda^j_{n-1}}a^q_n))\\
	+&(-1)^{i+j-1}\gamma_{\cdots,\lambda^q_n+\lambda^i_{1\to n-1}+\lambda^j_{1\to n-1}}(\cdots,{L(\overrightarrow{a^j})}_{\lambda^j_{n-1}}({L(\overrightarrow{a^i})}_{\lambda^i_{n-1}}a^q_n)).
	\end{align*} And this exhausts terms of $S_{22}$. By (\ref{12}) and linearity of $\gamma$, the sum is zero for each $i<j$.
	
	\textbf{($S_{13}$, $S_{31}$ and $S_{33}$).} Terms in $S_{13}$ have two types.\begin{align*}
	&(a)\ {a^i_1}_{\lambda^i_1}{\cdots}_{\lambda^i_{n-2}}{a^i_{n-1}}_{\lambda^i_{n-1}}(\gamma_{\cdots,[\lambda^j_{1\to n-1},\lambda^k_{1\to n-1}],\cdots,\lambda^q_n}(\cdots,[{a^j_{1\to n-1}},a^k_{1\to n-1}],\cdots,a^q_n)),\\
	&(b)\ [{L(\overrightarrow{a^i})}_{\lambda^i_{n-1}}L(\overrightarrow{a^j})]_{\lambda^i_{n-1}+\lambda^j_{n-1}}(\gamma_{\lambda^{1\to\hat{i}\to\hat{j}\to q}_{1\to n-1},\lambda^q_n}(\cdots,\hat{a}^i_{1\to n-1},\cdots,\hat{a}^j_{1\to n-1},\cdots,a^q_n)).
	\end{align*}
	
	Terms of type (a) in $S_{13}$ and terms of $S_{31}$ are in the same form. But they have opposite signs in different summands whether $i<j$ or $i>j$ for fixed $i,j,k$. So they are offset. On the other hand, each term of type (b) can be matched with two terms of $S_{33}$ as follows.\begin{align*}
	&(-1)^{i+j}[{L(\overrightarrow{a^i})}_{\lambda^i_{n-1}}L(\overrightarrow{a^j})]_{\lambda^i_{n-1}+\lambda^j_{n-1}}(\gamma_{\lambda^{1\to\hat{i}\to\hat{j}\to q}_{1\to n-1},\lambda^q_n}(\cdots,\hat{a}^i_{1\to n-1},\cdots,\hat{a}^j_{1\to n-1},\cdots,a^q_n))\\
	+&(-1)^{i+j+1}{L(\overrightarrow{a^i})}_{\lambda^i_{n-1}}({L(\overrightarrow{a^j})}_{\lambda^j_{n-1}}(\gamma_{\lambda^{1\to\hat{i}\to\hat{j}\to q}_{1\to n-1},\lambda^q_n}(\cdots,\hat{a}^i_{1\to n-1},\cdots,\hat{a}^j_{1\to n-1},\cdots,a^q_n)))\\
	+&(-1)^{i+j+2}{L(\overrightarrow{a^j})}_{\lambda^j_{n-1}}({L(\overrightarrow{a^i})}_{\lambda^i_{n-1}}(\gamma_{\lambda^{1\to\hat{i}\to\hat{j}\to q}_{1\to n-1},\lambda^q_n}(\cdots,\hat{a}^i_{1\to n-1},\cdots,\hat{a}^j_{1\to n-1},\cdots,a^q_n))).
	\end{align*} And this exhausts terms of $S_{33}$. By (M4) of $M$, the sum is zero for each $i<j$.
	
	\textbf{($S_{23}$ and $S_{32}$).} Terms of type (a) in $S_{23}$ and terms of $S_{32}$ are in the same form as follows.
	\begin{align*}
	{a^i_1}_{\lambda^i_1}{\cdots}_{\lambda^i_{n-2}}{a^i_{n-1}}_{\lambda^i_{n-1}}(\gamma_{\lambda^{1\to\hat{i}\to\hat{j}\to q}_{1\to n-1},\lambda^q_n+\lambda^{j}_{1\to n-1}}(\cdots,\hat{a}^i_{1\to n-1},\cdots,\hat{a}^j_{1\to n-1},\cdots,[{a^j_1}_{\lambda^j_1}{\cdots}{a^j_{n-1}}_{\lambda^j_{n-1}}a^q_n])).
	\end{align*} But they have opposite signs in different summands whether $i<j$ or $i>j$ for fixed $i,j$. So they cancel each other out. 
	
	\textbf{($S_{14}$, $S_{41}$, $S_{42}$, part of $S_{24}$, $S_{43}$ and part of $S_{34}$).} Since the expansion of $[{L(\overrightarrow{a^i})}_{\lambda^i_{n-1}}L(\overrightarrow{a^j})]_{\lambda^i_{n-1}+\lambda^j_{n-1}}$ is $$\sum_{l=1}^{n-1}({a^j_1} _{\lambda^j_1}\cdots{ a^j_{l-1}}_{\lambda^j_{l-1}}({L(\overrightarrow{a^i})}_{\lambda^i_{n-1}}a^j_l)_{\lambda^j_l+\lambda^i_{1\to n-1}}{ a^j_{l+1}}_{\lambda^j_{l+1}}\cdots{a^j_{n-1}})_{\lambda^j_{n-1}},$$
	terms in $S_{14}$ have the following forms.
	\begin{align*}
	&(a)\ {a^q_1}_{\lambda^q_1}\cdots{\hat{a}^q_k}\cdots{a^q_{n-1}}_{\lambda^q_{n-1}} {a^q_n}_{\lambda^q_n}(\gamma_{\cdots,[\lambda^i_{1\to n-1},\lambda^j_{1\to n-1}],\cdots,\lambda^q_k}(\cdots,[{a^i_{1\to n-1}},a^j_{1\to n-1}],\cdots, a^q_k)),\\
	&(b)\ {a^q_1}_{\lambda^q_1}\cdots{\hat{a}^q_j}\cdots{a^q_{n-1}}_{\lambda^q_{n-1}} {a^q_n}_{\lambda^q_n}(\gamma_{\lambda^{1\to\hat{i}\to q-1}_{1\to n-1},\lambda^q_j+\lambda^i_{1\to n-1}}(\cdots, \hat{a}^{i}_{1\to n-1},\cdots, ({L(\overrightarrow{a^i})}_{\lambda^i_{n-1}}a^q_j))),\\
	&(c)\ \cdots{\hat{a}^q_k}\cdots({L(\overrightarrow{a^i})}_{\lambda^i_{n-1}}a^q_j)_{\lambda^q_j+\lambda^i_{1\to n-1}}\cdots {a^q_n}_{\lambda^q_n}(\gamma_{\lambda^{1\to\hat{i}\to q-1}_{1\to n-1},\lambda^q_k+\lambda^i_{1\to n-1}}(\cdots, \hat{a}^{i}_{1\to n-1},\cdots, {a}^q_k)),\\
	&(d)\ \cdots({L(\overrightarrow{a^i})}_{\lambda^i_{n-1}}a^q_j)_{\lambda^q_j+\lambda^i_{1\to n-1}}\cdots{\hat{a}^q_k}\cdots {a^q_n}_{\lambda^q_n}(\gamma_{\lambda^{1\to\hat{i}\to q-1}_{1\to n-1},\lambda^q_k+\lambda^i_{1\to n-1}}(\cdots, \hat{a}^{i}_{1\to n-1},\cdots, {a}^q_k)).
	\end{align*}
	
	Terms of type (a) in $S_{14}$ and terms of $S_{41}$ are in the same form. But they have opposite signs in different summands for fixed $i,j,k$. So they are offset. Similarly, terms of type (b) in $S_{14}$ and part terms of $S_{42}$ cancel each other out for fixed $i<j=p$.
	
	Fix $i,k$, terms of type (c) and (d) plus corresponding terms in $S_{24}$\begin{align*}
	(-1)^{i+n+q-k}\cdots{\hat{a}^q_k}\cdots ({L(\overrightarrow{a^i})}_{\lambda^i_{n-1}}a^q_n)_{\lambda^q_n+\lambda^i_{1\to n-1}}(\gamma_{\lambda^{1\to\hat{i}\to q-1}_{1\to n-1},\lambda^q_k+\lambda^i_{1\to n-1}}(\cdots, \hat{a}^{i}_{1\to n-1},\cdots, {a}^q_k))	
	\end{align*} equal \begin{align*}
	(-1)^{i+n+q-k}[{L(\overrightarrow{a^i})}_{\lambda^i_{n-1}}L(\overrightarrow{a^q_k})]_{\lambda^i_{n-1}+\lambda^q_n}(\gamma_{\lambda^{1\to\hat{i}\to q-1}_{1\to n-1},\lambda^q_k+\lambda^i_{1\to n-1}}(\cdots, \hat{a}^{i}_{1\to n-1},\cdots, {a}^q_k)),
	\end{align*} where $L(\overrightarrow{a^q_k}):={a^q_1}_{\lambda^q_1}\cdots{\hat{a}^q_k}{\cdots}_{\lambda^q_{n-1}} {a^q_n}$. It can be matched with terms of $S_{43}$\begin{align*}
	(-1)^{i+n+q-k+2}{L(\overrightarrow{a^q_k})}_{\lambda^q_n}{L(\overrightarrow{a^i})}_{\lambda^i_{n-1}}((\gamma_{\lambda^{1\to\hat{i}\to q-1}_{1\to n-1},\lambda^q_k+\lambda^i_{1\to n-1}}(\cdots, \hat{a}^{i}_{1\to n-1},\cdots, {a}^q_k)))
	\end{align*} and part terms of $S_{34}$\begin{align*}
	(-1)^{i+n+q-k+1}{L(\overrightarrow{a^i})}_{\lambda^i_{n-1}}({L(\overrightarrow{a^q_k})}_{\lambda^q_n}(\gamma_{\lambda^{1\to\hat{i}\to q-1}_{1\to n-1},\lambda^q_k+\lambda^i_{1\to n-1}}(\cdots, \hat{a}^{i}_{1\to n-1},\cdots, {a}^q_k)))
	\end{align*}for fixed $i<j=q$. From (M4), they are offset.
	
	\textbf{($S_{44}$ and the rest of $S_{34}$ and $S_{24}$).} Adding up the remainders yields the following equation. \begin{align*}
	&\sum_{i=1}^{n-1}(-1)^{n-i}{a^{q-1}_1}_{\lambda^{q-1}_1}\cdots{\hat{a}^{q-1}_i}\cdots{a^{q-1}_{n-1}}_{\lambda^{q-1}_{n-1}} {[{a^q_1}_{\lambda^q_1}{\cdots}_{\lambda^q_{n-2}}{a^q_{n-1}}_{\lambda^q_{n-1}}a^q_n]}_{\lambda^q_n+\lambda^q_{1\to n-1}}\\
	&\qquad(\gamma_{\lambda^{1\to q-2}_{1\to n-1},\lambda^{q-1}_i}(a^1_{1\to n-1},\cdots, a^{q-2}_{1\to n-1}, a^{q-1}_i))\\
	&+\sum_{i=1}^{n-1}(-1)^{n-i+1}{a^q_1}_{\lambda^q_1}{\cdots}_{\lambda^q_{n-2}}{a^q_{n-1}}_{\lambda^q_{n-1}}({a^{q-1}_1}_{\lambda^{q-1}_1}\cdots{\hat{a}^{q-1}_i}\cdots{a^{q-1}_{n-1}}_{\lambda^{q-1}_{n-1}} {a^q_n}_{\lambda^q_n}\\
	&\qquad(\gamma_{\lambda^{1\to q-2}_{1\to n-1},\lambda^{q-1}_i}(a^1_{1\to n-1},\cdots, a^{q-2}_{1\to n-1}, a^{q-1}_i)))\\
	&+\sum_{i=1}^{n-1}\sum_{j=1}^{n-1}(-1)^{i+j+1}{a^q_1}_{\lambda^q_1}\cdots{\hat{a}^{q}_j}{\cdots}{a^q_{n-1}}_{\lambda^q_{n-1}}{a^q_n}_{\lambda^q_n}({a^{q-1}_1}_{\lambda^{q-1}_1}\cdots{\hat{a}^{q-1}_i}\cdots{a^{q-1}_{n-1}}_{\lambda^{q-1}_{n-1}} {a^q_j}_{\lambda^q_j}\\
	&\qquad(\gamma_{\lambda^{1\to q-2}_{1\to n-1},\lambda^{q-1}_i}(a^1_{1\to n-1},\cdots, a^{q-2}_{1\to n-1}, a^{q-1}_i)))\\
	=&\sum_{i=1}^{n-1}(-1)^{i}{[{a^q_1}_{\lambda^q_1}{\cdots}_{\lambda^q_{n-2}}{a^q_{n-1}}_{\lambda^q_{n-1}}a^q_n]}_{\lambda^q_n+\lambda^q_{1\to n-1}}{a^{q-1}_1}_{\lambda^{q-1}_1}\cdots{\hat{a}^{q-1}_i}\cdots{a^{q-1}_{n-1}}_{\lambda^{q-1}_{n-1}} \\
	&\qquad(\gamma_{\lambda^{1\to q-2}_{1\to n-1},\lambda^{q-1}_i}(a^1_{1\to n-1},\cdots, a^{q-2}_{1\to n-1}, a^{q-1}_i))\\
	&+\sum_{i=1}^{n-1}\sum_{j=1}^{n}(-1)^{n+i-j+1}{a^q_1}_{\lambda^q_1}\cdots{\hat{a}^{q}_j}{\cdots}{a^q_{n-1}}_{\lambda^q_{n-1}}{a^q_n}_{\lambda^q_n}({a^q_j}_{\lambda^q_j}{a^{q-1}_1}_{\lambda^{q-1}_1}\cdots{\hat{a}^{q-1}_i}\cdots{a^{q-1}_{n-1}}_{\lambda^{q-1}_{n-1}} \\
	&\qquad(\gamma_{\lambda^{1\to q-2}_{1\to n-1},\lambda^{q-1}_i}(a^1_{1\to n-1},\cdots, a^{q-2}_{1\to n-1}, a^{q-1}_i)))
	\end{align*} By the first equation of (M4), the value is zero.
	
	Therefore, $D^2=0$.
 \end{proof}
Hence $(\widetilde{C}^*(R,M):=\{\widetilde{C}^q(R,M),q\ge0\},D)$ is a complex. Each $q$-cochain $\widetilde{C}^q(R,M)$ can be regarded as a $\mathbb{C}[\partial]$-module via\begin{align}
(\partial\gamma)_{\lambda^{1\to q-1}_{1\to n-1},\lambda^{q-1}_n}(a^1_{1\to n-1},&\cdots, a^{q-1}_{1\to n-1}, a^{q-1}_n)\notag\\
&=(\partial+\lambda^{1\to q-1}_{1\to n-1}+\lambda^{q-1}_n)\gamma_{\lambda^{1\to q-1}_{1\to n-1},\lambda^{q-1}_n}(a^1_{1\to n-1},\cdots, a^{q-1}_{1\to n-1}, a^{q-1}_n).
\end{align} 
\begin{corollary}
	$D\partial=\partial D$. Thus $\partial\widetilde{C}^*(R,M)$ is a subcomplex of $\widetilde{C}^*(R,M)$.
\end{corollary}
\begin{proof}
	Consider the expansion of $(D(\partial\gamma))_{\lambda^{1\to q}_{1\to n-1},\lambda^q_n}(a^1_{1\to n-1},\cdots, a^q_{1\to n-1}, a^q_n)$ by (\ref{16}). We can easily extract $\partial+\lambda^{1\to q}_{1\to n-1}+\lambda^{q}_n$ from the first two summands. And $\partial+\lambda^{1\to q}_{1\to n-1}+\lambda^{q}_n$ of the last two summands can be deduced by applying (M2). 
\end{proof}
 Then we give the following definitions.
\begin{definition}
	\begin{em}
		The complex $\widetilde{C}^*(R,M)$ endowed with differential $D$ is called the \textit{basic complex of $R$ with coefficients in $M$}. And the quotient complex $C^*(R,M)=\widetilde{C}^*(R,M)/\partial\widetilde{C}^*(R,M)$ is called the \textit{reduced complex of $R$ with coefficients in $M$}.
	\end{em}
\end{definition}

Recall the complex of an $n$-Lie algebra $\mathfrak{g}$ with coefficients in a $\mathfrak{g}$-representation $(V,\rho)$ (See, e.g.,\cite{lszc}) is composed by 

(i) $q$-cochains $C^q(\mathfrak{g},V):=\{\gamma:\mathfrak{g}^{\wedge^{n-1}}\otimes{\tiny \begin{matrix}
(q-1)\\\cdots
\end{matrix}} \otimes\mathfrak{g}^{\wedge^{n-1}}\wedge \mathfrak{g}\to V\}$,

(ii) differential $D:C^q(\mathfrak{g},V)\to C^{q+1}(\mathfrak{g},V)$ given by\begin{align*}
&(D\gamma)(a^1_{1\to n-1},\cdots, a^q_{1\to n-1}, a^q_n)\notag\\
=&\sum_{1\le i<j\le q}(-1)^{i}\gamma(a^1_{1\to n-1},\cdots,\hat{a}^i_{1\to n-1},\cdots,a^{j-1}_{1\to n-1},[{a^i_{1\to n-1}},a^j_{1\to n-1}],a^{j+1}_{1\to n-1},\cdots,a^q_n)\notag\\
&+\sum_{i=1}^{q}(-1)^{i}\gamma(a^1_{1\to n-1},\cdots,\hat{a}^i_{1\to n-1},\cdots,{a}^q_{1\to n-1},[{a^i_1},{\cdots},{a^i_{n-1}},a^q_n])\notag\\
&+\sum_{i=1}^{q}(-1)^{i+1}\rho({a^i_1},{\cdots},{a^i_{n-1}})(\gamma(a^1_{1\to n-1},\cdots,\hat{a}^i_{1\to n-1},\cdots,{a}^q_{1\to n-1},a^q_n))\notag\\
&+\sum_{i=1}^{n-1}(-1)^{n+q-i+1}\rho({a^q_1},\cdots,{\hat{a}^q_i},\cdots,{a^q_{n-1}}, {a^q_n})(\gamma(a^1_{1\to n-1},\cdots, a^{q-1}_{1\to n-1}, a^q_i)),
\end{align*}where 
$[{a^i_{1\to n-1}},a^j_{1\to n-1}]=\sum_{l=1}^{n-1}\limits ({a^j_1} ,\cdots,{ a^j_{l-1}},[{a^i_1},{\cdots},{a^i_{n-1}},a^j_l],{ a^j_{l+1}},\cdots,{a^j_{n-1}})$.

Let \begin{align*}
\gamma&_{\lambda^{1\to q-1}_{1\to n-1},\lambda^{q-1}_n}(a^1_{1\to n-1},\cdots, a^{q-1}_{1\to n-1}, a^{q-1}_n)\\
&=\sum_{m^{1\to q-1}_{1\to n-1},m^{q-1}_n}{\lambda^1_1}
^{(m^1_1)}\cdots{\lambda^i_j}^{(m^i_j)}\cdots{\lambda^{q-1}_n}^{(m^{q-1}_n)}\gamma_{(m^{1\to q-1}_{1\to n-1},m^{q-1}_n)}(a^1_{1\to n-1},\cdots, a^{q-1}_{1\to n-1}, a^{q-1}_n),
\end{align*} and we can rephrase the differential of the basic complex as\begin{align*}
&(D\gamma)_{(m^{1\to q}_{1\to n-1},m^{q}_n)}(a^1_{1\to n-1},\cdots, a^q_{1\to n-1}, a^q_n)\notag\\
=&\sum_{1\le i<j\le q}(-1)^{i}\gamma_{(m^{1\to\hat{i}\to j-1}_{1\to n-1},[m^i_{1\to n-1},m^j_{1\to n-1}],m^{j+1\to q}_{1\to n-1},m^q_n)}(a^1_{1\to n-1},\cdots,\hat{a}^i_{1\to n-1},\cdots,a^{j-1}_{1\to n-1},\notag
\\&\qquad\qquad\qquad\qquad\qquad\qquad\qquad\qquad\qquad\qquad\qquad[{a^i_{1\to n-1}},a^j_{1\to n-1}],a^{j+1}_{1\to n-1},\cdots,a^q_n)\notag\\
&+\sum_{i=1}^{q}\sum_{k^i_1,\cdots,k^i_{n-1}}(-1)^{i}\tbinom{m^i_1}{k^i_1}\cdots\tbinom{m^i_{n-1}}{k^i_{n-1}}\gamma_{(m^{1\to\hat{i}\to q}_{1\to n-1},m^q_n+m^i_{1\to n-1}-k^i_{1\to n-1})}(a^1_{1\to n-1},\cdots,\hat{a}^i_{1\to n-1},\notag
\\&\qquad\qquad\qquad\qquad\qquad\qquad\qquad\qquad\qquad\qquad\qquad\cdots,{a}^q_{1\to n-1},{a^i_1}_{(k^i_1)}{\cdots}_{(k^i_{n-2})}{a^i_{n-1}}_{(k^i_{n-1})}a^q_n)\notag\\
&+\sum_{i=1}^{q}(-1)^{i+1}{a^i_1}_{(m^i_1)}{\cdots}_{(m^i_{n-2})}{a^i_{n-1}}_{(m^i_{n-1})}(\gamma_{(m^{1\to\hat{i}\to q}_{1\to n-1},m^q_n)}(a^1_{1\to n-1},\cdots,\hat{a}^i_{1\to n-1},\cdots,{a}^q_{1\to n-1},a^q_n))\notag\\
&+\sum_{i=1}^{n-1}(-1)^{n+q-i+1}{a^q_1}_{(m^q_1)}\cdots{\hat{a}^q_i}\cdots{a^q_{n-1}}_{(m^q_{n-1})} {a^q_n}_{(m^q_n)}(\gamma_{(m^{1\to q-1}_{1\to n-1},m^q_i)}(a^1_{1\to n-1},\cdots, a^{q-1}_{1\to n-1}, a^q_i)),
\end{align*} where $[m^i_{1\to n-1},m^j_{1\to n-1}]$ and $[{a^i_{1\to n-1}},a^j_{1\to n-1}]$ correspond to\begin{align*}
&\sum_{l=1}^{n-1}\sum_{k^i_1,\cdots,k^i_{n-1}}\tbinom{m^i_1}{k^i_1}\cdots\tbinom{m^i_{n-1}}{k^i_{n-1}}({a^j_1} _{(m^j_1)}\cdots{ a^j_{l-1}}_{(m^j_{l-1})}\\&\qquad\qquad\qquad({a^i_1}_{(k^i_1)}{\cdots}_{(k^i_{n-2})}{a^i_{n-1}}_{(k^i_{n-1})}a^j_l)_{(m^j_l+m^i_{1\to n-1}-k^i_{1\to n-1})}{ a^j_{l+1}}_{(m^j_{l+1})}\cdots{a^j_{n-1}})_{(m^j_{n-1})}.
\end{align*}  Obviously, the subcomplex, whose $q$-cochain is made of $\gamma_{(0,\cdots,0)}$, is just the complex of $n$-Lie algebra $R$ with coefficients in $M$. Denote the complex of $n$-Lie algebra $(\mathscr{L}ie_p\mbox{ }R)_\_$ with coefficients in $M$ by $C^*((\mathscr{L}ie_p\mbox{ }R)_\_,M)$. $C^q((\mathscr{L}ie_p\mbox{ }R)_\_,M)$ carries a $\mathbb{C}[\partial]$-module structure defined by \begin{align*}
&(\partial\gamma)({a^1_1}_{m^1_{1\ 1\to p}},\cdots,{a^i_j}_{m^i_{j\ 1\to p}},\cdots,{a^{q-1}_n}_{m^{q-1}_{n\ 1\to p}})\\
=&\partial(\gamma({a^1_1}_{m^1_{1\ 1\to p}},\cdots,{a^i_j}_{m^i_{j\ 1\to p}},\cdots,{a^{q-1}_n}_{m^{q-1}_{n\ 1\to p}}))\\
&-\sum_{i,j}\gamma({a^1_1}_{m^1_{1\ 1\to p}},\cdots,{\partial a^i_j}_{m^i_{j\ 1\to p}},\cdots,{a^{q-1}_n}_{m^{q-1}_{n\ 1\to p}}).
\end{align*} And similar to \cite{bkv}, $\widetilde{C}^*(R,M)$ and $C^*((\mathscr{L}ie_p\mbox{ }R)_\_,M)$ are closely related.
\begin{theorem}\label{t1}
	For $\gamma\in \widetilde{C}^*(R,M)$, define 
	$$(\varphi\gamma)({a^1_1}_{m^1_{1\ 1\to p}},\cdots,{a^{q-1}_n}_{m^{q-1}_{n\ 1\to p}})=\gamma_{(\sum m^1_{1\ 1\to p},\cdots,\sum m^{q-1}_{n\ 1\to p})}({a^1_1},\cdots,{a^{q-1}_n}).$$
	Then $\varphi:\widetilde{C}^*(R,M)\to C^*((\mathscr{L}ie_p\mbox{ }R)_\_,M)$ is a monomorphism of cochain complexes compatible with the $\partial$-actions for all $p\ge1$. When $p=1$, $\varphi$ is an isomorphism.
\end{theorem}
\begin{proof}
	(B2) implies $\varphi\gamma\in C^*((\mathscr{L}ie_p\mbox{ }R)_\_,M)$. For $\gamma_1,\gamma_2\in\widetilde{C}^q(R,M)$, it is easy to see that $\varphi\gamma_1=\varphi\gamma_2$ if and only if $\gamma_1=\gamma_2$. $\gamma$ and $\varphi\gamma$ are in one-to-one correspondence when $p=1$. Since
	\begin{align*}
	&(\partial(\varphi\gamma))({a^1_1}_{m^1_{1\ 1\to p}},\cdots,{a^{q-1}_n}_{m^{q-1}_{n\ 1\to p}})\\
	=&\partial((\varphi\gamma)({a^1_1}_{m^1_{1\ 1\to p}},\cdots,{a^i_j}_{m^i_{j\ 1\to p}},\cdots,{a^{q-1}_n}_{m^{q-1}_{n\ 1\to p}}))\\
	&-\sum_{i,j}(\varphi\gamma)({a^1_1}_{m^1_{1\ 1\to p}},\cdots,{\partial a^i_j}_{m^i_{j\ 1\to p}},\cdots,{a^{q-1}_n}_{m^{q-1}_{n\ 1\to p}})\\
	=&\partial((\varphi\gamma)({a^1_1}_{m^1_{1\ 1\to p}},\cdots,{a^i_j}_{m^i_{j\ 1\to p}},\cdots,{a^{q-1}_n}_{m^{q-1}_{n\ 1\to p}}))\\
	&+\sum_{i,j,k}m^i_{j\ k}(\varphi\gamma)({a^1_1}_{m^1_{1\ 1\to p}},\cdots,{a^i_j}_{m^i_{j\ 1},\cdots,m^i_{j\ k}-1,\cdots,m^i_{j\ p}},\cdots,{a^{q-1}_n}_{m^{q-1}_{n\ 1\to p}})\\
	=&(\partial\gamma)_{(\sum m^1_{1\ 1\to p},\cdots,\sum m^{q-1}_{n\ 1\to p})}({a^1_1},\cdots,{a^{q-1}_n})\\
	=&(\varphi(\partial\gamma))({a^1_1}_{m^1_{1\ 1\to p}},\cdots,{a^{q-1}_n}_{m^{q-1}_{n\ 1\to p}}),
	\end{align*} $\varphi$ commutes with $\partial$. And $\varphi D\gamma=D\varphi\gamma$ can be checked by summands of the expansions. Take the $i$th term of the second summands as an example.
	
Comparing the coefficients of $x^k$ in the two sides of $(x+1)^{m_1+\cdots+m_p}=(x+1)^{m_1}\cdots(x+1)^{m_p}$, we have $$\tbinom{m_1+\cdots+m_p}{k}=\sum_{k_1+\cdots+k_p=k}\tbinom{m_1}{k_1}\cdots\tbinom{m_p}{k_p}.$$ Hence the $i$th term of $D\varphi \gamma$-action on $ ({a^1_1}_{m^1_{1\ 1\to p}},\cdots,{a^{q}_n}_{m^{q}_{n\ 1\to p}}) $ can be transformed into that of $\varphi D\gamma$ in the following way.\begin{align*}
	&(-1)^{i}(\varphi\gamma)({a^1_1}_{m^1_{1\ 1\to p}},\cdots,\hat{a}^i_{1\to n-1},\cdots,{a^{q}_{n-1}}_{m^{q}_{n-1\ 1\to p}},[{a^i_1}_{m^i_{1\ 1\to p}},{\cdots},{a^i_{n-1}}_{m^i_{n-1\ 1\to p}},{a^{q}_n}_{m^{q}_{n\ 1\to p}}])\\
	=&(-1)^{i}\sum_{k^i_{1\to n-1\ 1\to p}}\tbinom{m^i_{1\ 1}}{k^i_{1\ 1}}\cdots\tbinom{m^i_{n-1\ p}}{k^i_{n-1\ p}}\gamma_{(m^{1\to\hat{i}\to q}_{1\to n-1},m^q_n+m^i_{1\to n-1}-k^i_{1\to n-1\ 1\to p})}(a^1_{1\to n-1},\cdots,\hat{a}^i_{1\to n-1},\notag
	\\&\qquad\qquad\cdots,{a}^q_{1\to n-1},{a^i_1}_{(\sum k^i_{1\ 1\to p})}{\cdots}{a^i_{n-1}}_{(\sum k^i_{n-1\ 1\to p})}a^q_n)\\
	=&(-1)^{i}\sum_{k^i_1,\cdots,k^i_{n-1}}\sum_{\sum k^i_{1\ 1\to p}=k^i_1,\cdots,\sum k^i_{n-1\ 1\to p}=k^i_{n-1}}\tbinom{m^i_{1\ 1}}{k^i_{1\ 1}}\cdots\tbinom{m^i_{1\ p}}{k^i_{1\ p}}\cdots\tbinom{m^i_{n-1\ 1}}{k^i_{n-1\ 1}}\cdots\tbinom{m^i_{n-1\ p}}{k^i_{n-1\ p}}\\&\gamma_{(m^{1\to\hat{i}\to q}_{1\to n-1},m^q_n+m^i_{1\to n-1}-k^i_{1\to n-1})}(a^1_{1\to n-1},\cdots,\hat{a}^i_{1\to n-1},\notag
	\cdots,{a}^q_{1\to n-1},{a^i_1}_{(k^i_1)}{\cdots}_{(k^i_{n-2})}{a^i_{n-1}}_{(k^i_{n-1})}a^q_n)\\
	=&\sum_{k^i_1,\cdots,k^i_{n-1}}(-1)^{i}\tbinom{m^i_1}{k^i_1}\cdots\tbinom{m^i_{n-1}}{k^i_{n-1}}\gamma_{(m^{1\to\hat{i}\to q}_{1\to n-1},m^q_n+m^i_{1\to n-1}-k^i_{1\to n-1})}(a^1_{1\to n-1},\cdots,\hat{a}^i_{1\to n-1},\cdots,\notag
	\\&\qquad\qquad{a}^q_{1\to n-1},{a^i_1}_{(k^i_1)}{\cdots}_{(k^i_{n-2})}{a^i_{n-1}}_{(k^i_{n-1})}a^q_n).
	\end{align*}
	
\end{proof}
As a result, we obtain an isomorphism on cohomology $\widetilde{H}^*(R,M)\simeq H^*((\mathscr{L}ie_1\mbox{ }R)_\_,M)$.
\begin{remark}
	From the proof of Proposition \ref{p5}, one can find there are more than one cohomology complex with the same $q$-cochains in the basic complex. For example, if we only reserve the first summand of (\ref{16}), $D$ is still a well-defined differential. These complexes of $n$-Lie conformal algebra $R$ are matched with corresponding complexes of $n$-Lie algebra $(\mathscr{L}ie_p\mbox{ }R)_\_$ introduced in \cite{ai1} in the same way to Theorem \ref{t1}.
\end{remark}

	\end{document}